\documentclass[12pt]{amsart}

\usepackage{a4wide}
\usepackage{amssymb}
\usepackage{amsmath}
\usepackage{amsfonts}
\usepackage{amsthm}
\usepackage{graphicx}
\usepackage{epsfig}
\usepackage{longtable}
\usepackage{hyperref}
\usepackage[usenames,dvipsnames]{color}
\usepackage[normalem]{ulem}

\newtheorem{thm}{Theorem}
\newtheorem{lem}[thm]{Lemma}
\newtheorem{cor}[thm]{Corollary}

\newtheorem{prop}[thm]{Proposition}

\theoremstyle{definition}
\newtheorem{defn}[thm]{Definition}
\newtheorem{eg}[thm]{Example}

\theoremstyle{remark}
\newtheorem*{rmk}{Remark}

\newcommand{\eps}{\varepsilon}

\newcommand{\DEF}{\,{:=}\,}
\newcommand{\FED}{\,{=:}\,}

\newcommand{\PT}[1]{\mathbf{#1}}

\newcommand{\sgnEqconst}{G}

\newcommand{\dd}{\,d}

\DeclareMathOperator{\bal}{Bal}

\DeclareMathOperator{\betafcn}{B}

\DeclareMathOperator{\CAP}{cap}

\DeclareMathOperator{\essinf}{``\inf''}
\DeclareMathOperator{\gammafcn}{\Gamma}

\DeclareMathOperator{\dist}{dist}

\DeclareMathOperator{\kelvin}{K}
\DeclareMathOperator{\kelvinMEAS}{\mathcal{K}}

\DeclareMathOperator{\supp}{supp}

\newcommand{\PSET}{X}

\DeclareMathOperator{\HyperF}{F}
\DeclareMathOperator{\HyperTildeF}{\tilde{F}}
\newcommand{\Hypergeom}[5]{{\sideset{_#1}{_#2}\HyperF\!\left(\substack{\displaystyle#3\\\displaystyle#4};#5\right)}}
\newcommand{\HypergeomReg}[5]{{\sideset{_#1}{_#2}
\HyperTildeF\!\left(\substack{\displaystyle#3\\\displaystyle#4};#5\right)}}

\newcommand{\Pochhsymb}[2]{{\left(#1\right)_{#2}}}

\allowdisplaybreaks[1]

\title[Riesz external field problems on the hypersphere]{Riesz external field problems on the hypersphere and optimal point separation}
\author[J. S. Brauchart, P. D. Dragnev, and E. B. Saff]{Johann S. Brauchart\textasteriskcentered, Peter D. Dragnev\textdagger,
and Edward B. Saff\textdaggerdbl} 
\thanks{\noindent \textasteriskcentered The research of this author
was supported, in part, by an APART-Fellowship of the Austrian Academy of
Sciences and by an Australian Research Council Discovery grant. \\
\textdagger The research of this author was supported, in part, by a Grants-in-Aid program of ORESP at IPFW and by a grant from the Simons Foundation no. 282207.\\
\textdaggerdbl The research of this author was supported, in
part, by the U. S. National Science Foundation under grant DMS-1109266 as well as by an Australian Research Council Discovery grant.}

\date{\today}

\begin{document}

\address{J. S. Brauchart:
School of Mathematics and Statistics,
University of New South Wales,
Sydney, NSW, 2052,
Australia }
\address{P. D. Dragnev:
Department of Mathematical Sciences, Indiana University-Purdue
University Fort Wayne, Fort Wayne, IN 46805, USA}
\address{E. B. Saff:
Center for Constructive Approximation, Department of Mathematics,
Vanderbilt University, Nashville, TN 37240, USA}

\email{j.brauchart@unsw.edu.au} \email{dragnevp@ipfw.edu}
\email{Edward.B.Saff@Vanderbilt.Edu}

\begin{abstract}

We consider the minimal energy problem on the unit sphere $\mathbb{S}^d$ in the Euclidean space $\mathbb{R}^{d+1}$ in the presence of an external field $Q$, where the energy arises from the Riesz potential $1/r^s$ (where $r$ is the Euclidean distance and $s$ is the Riesz parameter) or the logarithmic potential $\log(1/r)$. Characterization theorems of Frostman-type for the associated extremal measure, previously obtained by the last two authors, are extended to the range $d-2  \leq s < d - 1.$  The proof uses a maximum principle for measures supported on $\mathbb{S}^d$. When $Q$ is the Riesz $s$-potential of a signed measure and $d-2 \leq s <d$,  our results lead to explicit point-separation estimates for $(Q,s)$-Fekete points, which are $n$-point configurations minimizing the Riesz $s$-energy on $\mathbb{S}^d$ with external field $Q$. In the hyper-singular case $s > d$, the short-range pair-interaction enforces well-separation even in the presence of more general external fields. As a further application,
we determine the extremal and signed equilibria when the external field is due to a negative point charge outside a positively charged isolated sphere. Moreover, we provide a rigorous analysis of the three point external field problem and numerical results for the four point problem.
\end{abstract}


\maketitle

\nocite{Zo2004}

\section{Introduction}

Let $\mathbb{S}^d \DEF \{ \PT{x} \in \mathbb{R}^{d+1} : |\PT{x}|=1
\}$ be the unit sphere in $\mathbb{R}^{d+1}$, where $|\PT{\cdot}|$
denotes the Euclidean norm. Given a compact set $E\subset
\mathbb{S}^d$, consider the class $\mathcal{M}(E)$ of unit positive
Borel measures supported on $E$. For $s>0$ the {\em Riesz
$s$-potential} and {\em Riesz $s$-energy} of a measure $\mu \in
\mathcal{M}(E)$ are given, respectively, by
\begin{equation*}
U_s^\mu(\PT{x}) \DEF \int k_s( \PT{x}, \PT{y} ) \dd \mu(\PT{y}),
\quad \PT{x} \in \mathbb{R}^{d+1}, \qquad \mathcal{I}_s(\mu) \DEF
\int \int k_s( \PT{x}, \PT{y} ) \dd \mu(\PT{x}) \dd \mu(\PT{y}),
\end{equation*}
where $k_s (\PT{x}, \PT{y}) \DEF |\PT{x}-\PT{y}|^{-s}$ is the
so-called {\em Riesz kernel}.
The {\em $s$-capacity} of $E$ is then defined as ${\rm cap}_s(E) \DEF 1/ W_s(E)$ for $s>0$, where $W_s(E) \DEF \inf \{ \mathcal{I}_s(\mu) : \mu \in \mathcal{M}(E) \}$ is the {\em $s$-energy} of the set $E$. A property is said to hold {\em quasi-everywhere} (q.e.), if the exceptional set has $s$-capacity zero.
When ${\rm cap}_s(E)>0$, there exists a unique minimizer $\mu_E = \mu_{s,E}$, called the {\em $s$-equilibrium measure on $E$}, such that $\mathcal{I}_s(\mu_E) = W_s(E)$.
For more details see \cite[Chapter~II]{La1972}.

Whenever $s = 0$ (we shall use $s = \log$), which occurs, for example, when $s = d - 2$ and $d = 2$, we replace the Riesz kernel $k_s$ by the {\em logarithmic kernel}
\begin{equation*}
k_{\log}(\PT{x}, \PT{y}) \DEF \log(1/|\PT{x}-\PT{y}|).
\end{equation*}
(In this case we define $\CAP_{\log}(E) \DEF \exp\{-W_{\log}(E)\}$.)

We shall refer to a lower semi-continuous function $Q:\mathbb{S}^d
\to (-\infty,\infty]$ such that ${Q(\PT{x})<\infty}$ on a set of
positive Lebesgue surface measure, as an {\em external field}. We
note that the lower semi-continuity implies the existence of a
finite $c_Q$ such that $Q(\PT{x})\geq c_Q$ for all $\PT{x}\in
\mathbb{S}^d$. The {\em weighted energy associated with $Q(\PT{x})$}
is then given by
\begin{equation} \label{energy}
I_{Q,s}(\mu) \DEF \mathcal{I}_s(\mu) + 2 \int Q(\PT{x}) \dd \mu(\PT{x}), \qquad \mu \in \mathcal{M}(E).
\end{equation}
(The terminology ``weighted energy'' is used here to indicate the presence of an external field, and should not be confused with ``weighted energy functionals'', where the Riesz $s$-kernel is multiplied by a weight function $w(\PT{x},\PT{y})$. We leave the study of the external field problem for such generalized kernels for a future investigation.)

\begin{defn} \label{def:external.field.problem}
The Riesz external field problem on the unit sphere $\mathbb{S}^d$
for the external field $Q$ is concerned with minimizing the weighted energy
\eqref{energy} among all Borel probability measures $\mu$ supported
on $\mathbb{S}^d$.
A measure $\mu_{Q,s} \in \mathcal{M}(\mathbb{S}^d)$ with
\begin{equation*}
I_{Q,s}(\mu_{Q,s}) = V_{Q,s} \DEF \inf \left\{ I_{Q,s}(\mu) : \mu \in \mathcal{M}(\mathbb{S}^d) \right\}
\end{equation*}
is called an {\em $s$-extremal (or positive equilibrium) measure on $\mathbb{S}^d$ associated with $Q$}.
\end{defn}

If we consider only measures supported on some compact subset $E \subset \mathbb{S}^d$ with positive $s$-capacity, then the minimizing measure is referred to as the {\em $s$-extremal measure on $E$ associated with $Q$} and denoted by $\mu_{E,Q,s}$. In the particular case when $Q\equiv 0$, the measure $\mu_{Q,s}$ on $\mathbb{S}^d$ is just the normalized unit surface area measure on the sphere for which we use the symbol $\sigma_d$.

We shall also consider the discrete analogue of the above external field problem which is defined as follows.
\begin{defn} \label{def:discrete.external.field.problem}
Let $s > 0$ or $s = \log$. For a set of $n$ points $\PSET_n = \{\PT{x}_1, \dots, \PT{x}_n \} \subset \mathbb{S}^d$ the {\em discrete weighted energy associated with $Q$} is given by
\begin{equation} \label{eq:discrete.Riesz.energy}
E_s^{Q}(\PSET_n) \DEF \mathop{\sum_{j=1}^n \sum_{k=1}^n}_{k \neq j} \Big[ k_s( \PT{x}_j, \PT{x}_k ) + Q(\PT{x}_j) + Q(\PT{x}_k) \Big].
\end{equation}
Then the discrete external field problem on the sphere $\mathbb{S}^d$ concerns the minimization
\begin{equation} \label{eq:discrete.Riesz.energy.minimum}
\mathcal{E}^Q_s (n) \DEF \min \Big\{ E_s^{Q}(\PSET_n) : \PSET_n \subset \mathbb{S}^d, | \PSET_n | = n \Big\},
\end{equation}
where $| A |$ denotes the cardinality of the set $A$.
A solution of the discretized minimization problem \eqref{eq:discrete.Riesz.energy.minimum} is called an {\em $n$-point $(Q,s)$-Fekete set}.
\end{defn}

The existence of $(Q,s)$-Fekete sets is an easy consequence of the lower semi-continuity of the energy functional and the compactness of the unit sphere. Further, we remark that a standard argument establishes the following monotonicity property
\begin{equation*}
\frac{\mathcal{E}^Q_s (n)}{n(n-1)} \leq \frac{\mathcal{E}^Q_s (n+1)}{(n+1)n} \qquad \text{for all $n \geq 2$.}
\end{equation*}

We remark that the discrete problem has application to image processing, namely the half-toning of images based on electrostatic repulsion of printed dots in the presence of an image-driven external field; cf. Schmaltz et al.~\cite{SchGwBrWei2010} and Gr{\"a}f~\cite[Section~6.5.2]{Gr2013}.

The outline of the paper is as follows.
In Section~\ref{sec:characterization} we provide Frostman-type characterization theorems for the solution to the external field minimal energy problem on the sphere. This is facilitated by a new restricted maximum principle on the sphere which holds for the range $d - 2 \leq s < d$ (see Theorem~\ref{thm:restr.max.principle}). We also introduce the signed equilibrium measure and discuss its relation to the positive equilibrium measure.
In Section~\ref{sec:separation} we establish that for a large class of external fields $Q$, the sequences of $n$-point $(Q,s)$-Fekete sets are well-separated; that is, have separation distance of order $n^{-1/d}$ (Theorems~\ref{thm:main3} and \ref{thm:main4}). 
In Section~\ref{sec:neg.external.field}, for an external field due to a negative point charge, we provide a detailed analysis and give explicit representations of the signed equilibrium (Theorem~\ref{thm:SignEq}) and  the $s$-extremal measure on $\mathbb{S}^d$ (Theorem~\ref{thm:s.equilibrium.measure}). This extends results in \cite{BrDrSa2009}.
In Section~\ref{sec:Examples} we rigorously characterize the $3$-point $(Q,s)$-Fekete set for a general class of convex external fields and provide numerical results for the four point problem with Riesz external fields (Figures~\ref{fig3} and \ref{fig2} illustrate the analysis).
The proofs of our results are provided in Section~\ref{sec:proofs}.

\section{Basic Properties and Characterization Theorems}
\label{sec:characterization}

In \cite{DrSa2007} the second and the third authors formulated the
following Frostman-type proposition, which deals with the existence
and uniqueness of the measure $\mu_{Q,s}$, as well as a criterion
that characterizes $\mu_{Q,s}$ in terms of its potential. The proof
of this proposition follows closely the proof of \cite[Theorem
I.1.3]{SaTo1997}. It could also be derived as a particular case from
the more general results in \cite{Zo2003} (see especially Theorems~1 and 2, and Proposition~1 of that paper).

\begin{prop} \label{prop:1}
Let $0< s < d$.\footnote{A similar result holds for the logarithmic case.} For the minimal energy problem on $\mathbb{S}^d$
with external field $Q$ the following properties hold:
\begin{itemize}
\item[\rm (a)] $V_{Q,s}$ is finite.
\item[\rm (b)] There exists a unique $s$-extremal measure
$\mu_{Q,s}\in \mathcal{M}(\mathbb{S}^d)$ associated with $Q$.
Moreover, the support $S_{Q,s} \DEF \supp( \mu_{Q,s} )$ of this measure is contained in the
compact set $E_M \DEF \{ \PT{x} \in \mathbb{S}^d : Q(\PT{x}) \leq M
\}$ for some $M>0$.
\item[\rm (c)] The measure $\mu_{Q,s}$ satisfies the variational inequalities
\begin{align}
U_s^{\mu_{Q,s}}(\PT{x}) + Q(\PT{x}) &\geq F_{Q,s} \quad \text{q.e. on $\mathbb{S}^d$,} \label{VarEq1} \\
U_s^{\mu_{Q,s}}(\PT{x}) + Q(\PT{x}) &\leq F_{Q,s} \quad
\text{everywhere on $S_{Q,s}$,} \label{VarEq2}
\end{align}
where
\begin{equation}
F_{Q,s} \DEF V_{Q,s} - \int Q(\PT{x}) \dd \mu_{Q,s}(\PT{x}).
\label{VarConst}
\end{equation}
\item[\rm (d)] Inequalities \eqref{VarEq1} and
\eqref{VarEq2} completely characterize the $s$-extremal measure
$\mu_Q$ in the sense that if $\nu \in \mathcal{M}(\mathbb{S}^d)$ is
a measure with finite $s$-energy such that for some constant $C$ we
have
\begin{align}
U_s^{\nu}(\PT{x}) + Q(\PT{x}) &\geq C \quad \text{q.e. on $\mathbb{S}^d$,} \label{VarEq3} \\
U_s^{\nu}(\PT{x}) + Q(\PT{x}) &\leq C \quad \text{everywhere on $\supp (\nu)$,} \label{VarEq4}
\end{align}
then $\nu=\mu_{Q,s}$ and $C=F_{Q,s}$.
\end{itemize}
\end{prop}

Observe that, if the external field $Q$ is continuous on $\mathbb{S}^d$, then the inequality in \eqref{VarEq3} holds everywhere on $\mathbb{S}^d$.

\begin{rmk}
Proposition~\ref{prop:1} remains true if $\mathbb{S}^d$ is replaced with any compact subset $K \subset \mathbb{S}^d$ with $\CAP_s( K ) > 0$. Notationally, the dependence on $K$ will be indicated by a subscript $K$ (e.g., $\mu_{K,Q,s}$, $F_{K,Q,s}$, etc.).
\end{rmk}

In the case when $d-1\leq s<d$, \cite[Theorem 1.3]{DrSa2007}
analyzes further the characterization property from
Proposition~\ref{prop:1}(d) by studying the supremum and the {\em
essential infimum} of the weighted potential $U_s^{\nu}(\PT{x}) +
Q(\PT{x})$. Our first theorem extends this analysis to the larger
range $d-2 \leq s < d$. 
To state the theorem we introduce the notation $\mathop{\essinf}_{\PT{x} \in E}$ to denote the {\em essential infimum} of $f$ with respect to a set $E\subset\mathbb{S}^d$; that is,
\begin{equation*}
\mathop{\essinf}_{\PT{x} \in E} f(\PT{x}) \DEF \sup \left\{ c : f(\PT{x}) \geq c \ \text{q.e. on $E$} \right\};
\end{equation*}
in other words, the infimum is taken quasi-everywhere.

\begin{thm} \label{thm:main}
Let $d-2 \leq s < d$, $Q$ be an external field on $\mathbb{S}^d$, and
$F_{Q,s}$ be defined as in \eqref{VarConst}. For any measure
$\lambda \in \mathcal{M}(\mathbb{S}^d)$ we have
\begin{align}
\mathop{\essinf}_{\PT{x} \in S_{Q,s}} \left[ U_s^\lambda(\PT{x}) +
Q(\PT{x}) \right] &\leq F_{Q,s} \label{eq:essinf} \intertext{and}
\sup_{\PT{x} \in\supp (\lambda)} \left[ U_s^\lambda(\PT{x}) + Q(\PT{x})
\right] &\geq F_{Q,s}. \label{eq:sup}
\end{align}
If equality holds in both inequalities, then $\lambda = \mu_{Q,s}$.
\end{thm}

For the restricted range $d-1\leq s < d$, the proof of this theorem as given in \cite{DrSa2007} utilizes the principle of domination for Riesz potentials, which generally is stated for the parameter range $d-1\leq s < d + 1$ and measures supported on any subsets of $\mathbb{R}^{d+1}$. (A restricted version of the principle of domination for $d-2 < s < d$ was established in \cite[Lemma~5.1]{DrSa2007}.)
Via a different approach that utilizes the following {\em restricted maximum principle on the sphere}, we are able to prove the result for $s$ in the extended range $d - 2 \leq s < d$; see Section~\ref{sec:proofs}. 

\begin{thm}[Sphere Maximum Principle] \label{thm:restr.max.principle}
Let {$d-2 \leq s < d$}. Suppose $\mu$ is a positive measure with $\supp (\mu) \subset \mathbb{S}^d$ such that for some $M > 0$, the relation $U_s^\mu(\PT{x}) \leq M$ holds $\mu$-almost everywhere on $\mathbb{S}^d$. Then $U_s^\mu(\PT{x}) \leq M$ holds everywhere on $\mathbb{S}^d$.
\end{thm}

An essential part of the analysis of external field problems is the determination the $s$-extremal (equilibrium) measure on $\mathbb{S}^d$ associated with the external field $Q$ and, in particular, its support. In principle, if the latter is known, the measure $\mu_{Q,s}$ can be recovered by solving an integral equation for the weighted $s$-potential of $\mu_{Q,s}$ arising from the variational inequalities \eqref{VarEq1} and \eqref{VarEq2}. A substantially easier problem is to find a (signed) measure that has constant weighted $s$-potential everywhere on $\mathbb{S}^d$. The solution of this problem turns out to be useful in solving the harder problem. This motivates the study of the signed equilibrium measure associated with an external field which is defined as follows.

\begin{defn} \label{def:signed.equilibrium}
Given a compact subset $K\subset \mathbb{R}^p$ ($p\geq 3$) and an external field $Q$, we call a signed measure $\eta_{K,Q}=\eta_{K,Q,s}$ supported on $K$ and of total charge $\eta_{K,Q}(K)=1$ {\em a signed $s$-equilibrium on $K$ associated with $Q$} if its weighted Riesz $s$-potential is constant on $K$; that is,
\begin{equation}  \label{signedeq}
U_s^{\eta_{K,Q}}(\PT{x}) + Q(\PT{x}) = \sgnEqconst_{K,Q,s} \qquad \text{for all $\PT{x} \in K$.}
\end{equation}
\end{defn}
We note that if a signed equilibrium exists, then it is unique (see \cite[Lemma 23]{BrDrSa2009}).

A remarkable connection exists to the Riesz analog of the \emph{Mhaskar-Saff $F$-functional} from classical logarithmic potential theory in the plane (see \cite{MhSa1985} and \cite[Chapter IV, p. 194]{SaTo1997}).
\begin{defn}
The \emph{$\mathcal{F}_s$-functional} of a compact subset $E \subset \mathbb{S}^d$ of positive $s$-capacity is defined as
\begin{equation} \label{Functional}
\mathcal{F}_s(E) \DEF W_s(E) + \int Q(\PT{x}) \, \dd \mu_E(\PT{x}),
\end{equation}
where $W_s(E)$ is the $s$-energy of $E$ and $\mu_E$ is the $s$-equilibrium measure (without external field) on $E$.
\end{defn}
Let $d - 2 \leq s < d$ with $s > 0$. If the signed equilibrium on a compact set $K \subset \mathbb{S}^d$ associated with $Q$ exists, then integration of \eqref{signedeq} with respect to $\mu_K$ shows that
\begin{equation} \label{eq:functional.identity}
\mathcal{F}_s( K ) = \sgnEqconst_{K,Q,s}.
\end{equation}
The essential property of the $\mathcal{F}_s$-functional is the following (cf. \cite[Theorem~9]{BrDrSa2009}).
\begin{prop} \label{prop:F.s.functional}
Let $d - 2 \leq s < d$ with $s > 0$ and $Q$ be an external field on a compact subset $K \subset \mathbb{S}^d$ with $\CAP_s( K ) > 0$. Then the $\mathcal{F}_s$-functional is minimized for the support of the $s$-extremal measure $\mu_{K,Q,s}$ on $K$ associated with $Q$; that is, for every compact subset $E \subset K$ with $\CAP_s( E ) > 0$,
\begin{equation*}
\mathcal{F}_s( E ) \geq \mathcal{F}_s( \supp( \mu_{K,Q,s} ) ) = F_{K,Q,s}. 
\end{equation*}
\end{prop}

Given a compact subset $K \subset \mathbb{S}^d$, the \emph{extended support $\widetilde{S}_{K,Q,s}$ of $\mu_{K,Q,s}$} is defined by
\begin{equation} \label{estlspt}
\widetilde{S}_{K,Q,s} \DEF \left\{ \PT{x} \in K : U_s^{\mu_{K,Q,s}}(\PT{x}) + Q(\PT{x}) \leq F_{K,Q,s} \right\}.
\end{equation}
The following theorem, which is the Riesz analog of \cite[Theorem~2.6]{DrSa1997} and \cite[Lemma 3]{KuDr1999}, establishes a relation between the extended support $\widetilde{S}_{Q,s}$ of $\mu_{Q,s}$ (by \eqref{VarEq2} this set contains the support of $\mu_{Q,s}$) and the support of the positive part $\eta_Q^+$ of the Jordan decomposition $\eta_Q^+ -\eta_Q^- $ of the signed equilibrium $\eta_Q=\eta_{\mathbb{S}^d,Q,s}$ on $\mathbb{S}^d$ associated with~$Q$.

\begin{thm}\label{thm:signsupp.2}
Let $d - 2 \leq s < d$ and suppose that $Q$ is an external field such that a signed $s$-equilibrium $\eta_\rho=\eta_{\Sigma_\rho,Q,s}$ on a spherical cap $\Sigma_\rho = \{ \PT{x} \in \mathbb{S}^d : | \PT{x} - \PT{p} | \geq \rho \}$ exists. Then
\begin{equation*} 
\mu_{Q,s} \big|_{\Sigma_\rho} \leq \eta^+_\rho \big|_{S_{Q,s}} \qquad \text{and} \qquad S_{Q,s} \cap \Sigma_\rho \subset \supp(\eta^+).
\end{equation*}
Furthermore, if $F_{Q,s} < \sgnEqconst_{\Sigma_\rho,Q,s}$, then $\widetilde{S}_{Q,s} \cap \Sigma_\rho \subset \supp ( \eta_\rho^+ )$.
\end{thm}

\begin{rmk}
The theorem remains true if the $s$-extremal measure on a compact subset $K \subset \mathbb{S}^d$ with $\CAP_s( K ) > 0$ and signed $s$-equilibria $\eta_E$  on compact subsets $E \subset \mathbb{S}^d$ with ${\supp( \eta_E^+ ) \subset K}$ are considered.
\end{rmk}

The characterization results for the shape of the support of the $s$-extremal measure on $\mathbb{S}^d$ associated with a rotational symmetric external field given in \cite[Theorem~10]{BrDrSa2009} immediately carry over to the external fields with extended range.
\begin{prop} \label{prop:ConnThm} Let $d-2\leq s<d$ with $s>0$ and the external field $Q:\mathbb{S}^d\to (-\infty,\infty]$ be rotationally invariant about the polar axis; that is, $Q(\PT{z})=f(\xi)$, where $\xi$ is the altitude of $\PT{z}=(\sqrt{1-\xi^2}\; \overline{\PT{z}},\xi)$, $\overline{\PT{z}} \in \mathbb{S}^{d-1}$. Suppose that $f$ is a convex function on $[-1,1]$. Then the support of the $s$-extremal measure $\mu_Q$ on $\mathbb{S}^d$ is a spherical zone; namely, there are numbers $-1\leq t_1\leq t_2\leq 1$ such that
\begin{equation} \label{eqsupp}
\supp(\mu_Q) = \Sigma_{t_1,t_2} \DEF \{ (\sqrt{1-u^2}\, \overline{\PT{x}},u)\ :\ t_1 \leq u \leq t_2, \, \overline{\PT{x}}\in \mathbb{S}^{d-1} \}.
\end{equation}
Moreover, if additionally $f$ is increasing, then $t_1=-1$ and the support of $\mu_Q$ is a spherical cap centered at the South Pole.
\end{prop}

Next we focus on the discretized version of the Riesz external field problem given in Definition~\ref{def:external.field.problem}.
Recall that the normalized counting measure associated with an $n$-point set $\PSET_n = \{
\PT{x}_1,\PT{x}_2,\dots,\PT{x}_n \}$ is defined as
\begin{equation*}
\mu_{\PSET_n} \DEF \frac{1}{n} \sum_{j=1}^n \delta_{\PT{x}_j},
\end{equation*}
where $\delta_{\PT{x}}$ is the Dirac-delta measure with unit mass at $\PT{x}$. The continuous and discrete external field minimization problems are related in the following way.
\begin{prop}
Let $0 < s < d$ or $s = \log$. Then
\begin{equation*}
\lim_{n \to \infty} \frac{\mathcal{E}^Q_s(n)}{n^2} = V_{Q,s} = I_{Q,s}(\mu_{Q,s}).
\end{equation*}
Furthermore, if $\{ \PSET_{n,Q,s} \}_{n=2}^\infty$ is any sequence of $n$-point $(Q,s)$-Fekete sets on $\mathbb{S}^d$ (see Definition~\ref{def:discrete.external.field.problem}), then the sequence of the normalized counting measures $\mu_{\PSET_{n,Q,s}}$ associated with $\PSET_{n,Q,s}$ converges in the weak-star sense to the $s$-extremal measure $\mu_{Q,s}$.
\end{prop}

The proof follows from a standard argument and utilizes the uniqueness result stated in Proposition~\ref{prop:1}(b).

We are interested in determining sets that contain all the $(Q,s)$-Fekete sets. For this purpose it is useful to investigate the \emph{weighted $s$-potential} of the normalized counting measure $\mu_{\PSET_n}$ which is defined as
\begin{equation} \label{DiscrPot}
h_{\PSET_n}(\PT{x}) \DEF U_s^{\mu_{\PSET_n}}(\PT{x}) + Q(\PT{x}) = \frac{1}{n} \sum_{j=1}^n
\frac{1}{\left|\PT{x}-\PT{x}_j\right|^s} + Q(\PT{x}), \qquad \PT{x} \in \mathbb{S}^d.
\end{equation}
As an application of Theorem~\ref{thm:main} we deduce the following result.

\begin{thm} \label{thm:main2}
Let $d-2 \leq s<d$. Let $\PSET_n \subset \mathbb{S}^d$ be a set of
$n$ distinct points, and suppose that, for some constant $M$, the
associated weighted potential satisfies the inequality
\begin{equation} \label{th7.1}
h_{\PSET_n} (\PT{x}) \geq M \qquad \text{q.e. on
$S_{Q,s}={\supp}(\mu_{Q,s})$}.
\end{equation}
Then (cf. \eqref{VarConst})
\begin{equation} \label{th7.2}
U^{\mu_{\PSET_n}}_s (\PT{x})\geq M + U^{\mu_{Q,s}}_s (\PT{x}) - F_{Q,s} \qquad \text{everywhere on $\mathbb{S}^d$.}
\end{equation}
Furthermore,
\begin{equation} \label{th7.3}
h_{\PSET_n}(\PT{x}) \geq M  \qquad \text{q.e. on $\mathbb{S}^d$.}
\end{equation}
\end{thm}

We point out that this is an extension of \cite[Theorem
1.7]{DrSa2007}, which, as with Theorem~\ref{thm:main} above, was
originally established for $d-1\leq s<d$. As in \cite[Corollary
1.9]{DrSa2007}, Theorems \ref{thm:main} and \ref{thm:main2} yield the following.

\begin{cor} \label{cor} 
For $d-2 \leq s < d$, every $(Q,s)$-Fekete set is contained in the extended support $\widetilde{S}_{Q,s}$.
\end{cor}

We note that for most of the above theorems, $s = d - 2$ marks the lower end of the stated range of the Riesz parameter $s$. It turns out that the case $s = d - 2$ is distinctive because new phenomena arise in the solution of the signed equilibrium problem, see Section~\ref{sec:neg.external.field}.
Moreover, for $s$ in the interval $(0, d-2)$, the Riesz-$s$ kernel becomes strictly superharmonic when considered in the stereographic projection space of $\mathbb{S}^d$; consequently maximum principles and domination principles do not apply.

\section{Application to point separation}
\label{sec:separation}

Good separation of points is generally associated with the {\em stability} of an approximation or interpolation method (e.g., by splines or radial basis functions (RBF)); cf., e.g., \cite{FuWr2009,LeGSlWe2010,Sch1995}.
In this section we shall apply results from Section~\ref{sec:characterization} (especially Theorem~\ref{thm:signsupp.2} and Corollary~\ref{cor}) to obtain explicit point separation estimates for sequences of $n$-point $(Q,s)$-Fekete sets (cf. Definition~\ref{def:discrete.external.field.problem}) associated with a large class of external fields $Q$ and establish that such sequences are ``well-separated'' in the following sense. Let
\begin{equation*}
\delta( \PSET_n ) \DEF \min \left\{ \left| \PT{x}_j - \PT{x}_k \right|
: \PT{x}_j, \PT{x}_k \in \PSET_n, j \neq k \right\}
\end{equation*}
denote the {\em minimum distance} among the points in $\PSET_n$. Then a sequence $\{\PSET_n\}_{n\geq2}$, $\PSET_n \subset \mathbb{S}^d$ for all $n$, is called {\em well-separated} if $\delta( \PSET_n )$ is of order $n^{-1/d}$ as $n \to \infty$. (It suffices to show the existence of a constant $C$ such that
\begin{equation} \label{eq:lower.estimate.least.distance}
\delta( \PSET_n ) \geq C \, n^{-1/d} \qquad \text{for sufficiently large $n$,}
\end{equation}
since $\delta( \PSET_n )$ cannot exceed the best-packing distance which is of order $n^{-1/d}$; cf. \cite{CoSl1999}.)

In the potential-theoretical and field-free setting ($Q\equiv 0$) it has been known since Dahlberg~\cite{Da1978} that Fekete point sets (harmonic case $s = d - 1$) on a sufficiently smooth closed bounded $d$-dimensional surface in $\mathbb{R}^{d+1}$ that separates $\mathbb{R}^{d+1}$ into two parts will form a well-separated sequence (but no explicit constant for the lower bound of $\delta( \PSET_n )$ has been given). G{\"o}tz~\cite{Go2000} studied the discrete external field problem on surfaces in $\mathbb{R}^d$ where the energy functional is defined in terms of the Green function for a domain $X \subset \mathbb{R}^d$. His separation result generalizes Dahlberg's result.
Well-separation of minimal logarithmic energy configurations on $\mathbb{S}^2$ in the field-free setting was first established by Rakhmanov et al.~\cite{RaSaZh1994,RaSaZh1995}, and with an improved constant by Dragnev~\cite{Dr2007a}.
For minimal Riesz $s$-energy configurations on $\mathbb{S}^d$ in the field-free case, well-separation was established by Kuijlaars et al.~\cite{KuSaSu2007} for $s \in (d-1,d)$ and by Dragnev and Saff~\cite{DrSa2007} for $s \in (d-2,d)$. Damelin and Maymeskul \cite{DaMa2005} give a separation result of order $n^{-1/(s+2)}$, $0 < s \leq d - 2$, which is of sharp order in the boundary case $s = d - 2$. It is expected but
still unproven that minimal logarithmic and Riesz $s$-energy ($0<s<d-2$) configurations on $\mathbb{S}^d$, $d \geq3$, are well-separated.  The references \cite{Dr2007a}, \cite{DrSa2007} and \cite{RaSaZh1995} also provide an explicit constant in the lower estimate
\eqref{eq:lower.estimate.least.distance}. It should be noted that \cite{DrSa2007} uses external fields to derive the desired separation  estimates in the field-free setting. In the hyper-singular case $s > d$, Kuijlaars and Saff~\cite{KuSa1998} establish well-separation of minimal Riesz $s$-energy configurations on $\mathbb{S}^d$.

We now present a generalization of \cite[Theorem~1.5]{DrSa2007} to the case when an external field is present and given by a potential.

\begin{thm}\label{thm:main3} Let $d-2 \leq s < d$ and $Q(\PT{x}) \DEF U_s^\sigma (\PT{x})$
for some signed measure $\sigma$ with $\sigma_d$-a.e.
finite Riesz $s$-potential on $\mathbb{S}^d$. Assume the support of the
negative part $\sigma^-$ in the Jordan decomposition
$\sigma = \sigma^+ -\sigma^-$ satisfies that
\begin{equation} \label{eq:sep.cond.1}
\supp( \sigma^- ) \subset \left\{ \PT{x} \in \mathbb{R}^{d+1} : |\PT{x}| \geq r \right\}
\end{equation}
for some $r>1$ and that
\begin{equation} \label{eq:sep.cond.2}
c_\sigma = c_\sigma(r) \DEF 1 + \| \sigma^+ \| + \left( \dfrac{(r+1)^{d-s}}{W_s( \mathbb{S}^d ) \, (r-1)^d} - 1 \right) \| \sigma^- \| \, \geq \frac{1}{2}. 
\end{equation}
Then any sequence $( \PSET_{n,Q,s} )_{n=2}^\infty$ of $(Q,s)$-Fekete sets on $\mathbb{S}^d$ is
well-separated; more precisely, 
\begin{equation} \label{SepRes1}
\delta( \PSET_{n,Q,s} ) \geq \frac{K_{Q,s}}{n^{1/d}} \qquad \text{for all $n > 2 c_\sigma - 1$,}
\end{equation}
where
\begin{equation} \label{SepRes2}
K_{Q,s} \DEF \left( \frac{2^{d-s}}{W_s( \mathbb{S}^d )} \frac{1}{c_\sigma} \right)^{1/d}.
\end{equation}
It is understood that for $d = 2$ and $s = \log$ we replace $W_s( \mathbb{S}^d )$ by $1$ and $s$ by $0$.
\end{thm}

Since the {\em Riesz $s$-energy of $\mathbb{S}^d$} appearing in the separation constant in \eqref{SepRes2} is given by the formula
\begin{equation} \label{eq:W.s.S.d}
W_s( \mathbb{S}^d ) = 2^{d-1-s} \frac{\gammafcn( ( d + 1 ) / 2 ) \gammafcn( ( d - s ) / 2 )}{\sqrt{\pi} \, \gammafcn( d - s / 2 )}, \qquad s > 0,
\end{equation}
we have in the harmonic case $s=d-1$ that
\begin{equation} \label{SepRes2b}
K_{Q,d-1} = 2^{-1/d} \left[ 1 + \| \sigma^+ \| + \left( \dfrac{r+1}{(r-1)^d}-1\right) \|\sigma^-\| \right]^{-1/d}
\end{equation}
and in the limiting case $s = d - 2$ (and $d \geq 3$)
\begin{equation} \label{SepRes2c}
K_{Q,d-2} = \left( \frac{1}{d} \, \frac{\gammafcn( ( d + 1 ) / 2 )}{\sqrt{\pi} \, \gammafcn( d / 2 )} \right)^{-1/d} \left[ 1 + \| \sigma^+ \| + \left( \dfrac{(r+1)^{2}}{W_{d-2}( \mathbb{S}^d ) \, (r-1)^d} - 1 \right) \| \sigma^- \| \right]^{-1/d},
\end{equation}
wheras for $s = \log$ and $d = 2$
\begin{equation} \label{SepRes2clog}
K_{Q,\log} = 2 \left[ 1 + \| \sigma^+ \| + \left( \dfrac{(r+1)^{2}}{(r-1)^2} - 1 \right) \| \sigma^- \| \right]^{-1/2}.
\end{equation}

\begin{rmk}
Note that whenever the support of $\sigma^-$ lies outside of $\mathbb{S}^d$, then both conditions~\eqref{eq:sep.cond.1} and \eqref{eq:sep.cond.2} are satisfied by taking $r(>1)$ sufficiently close to $1$. Also observe that as $r$ approaches $1$, the constant $K_{Q,s}$ approaches $0$.
\end{rmk}

In case of $\sigma \equiv 0$ and $s = d - 2 > 0$ the above Theorem~\ref{thm:main3} yields a known result for the well-separation of $n$-point minimal Riesz $(d-2)$-energy configurations on $\mathbb{S}^d$ (\cite{DaMa2005} but without explicit constants). Also with $\sigma \equiv 0$, $s = \log$ and $d = 2$ we recover the same separation result as obtained in \cite{Dr2007a}. Here we prove them separately (with explicit constants in the former case), since they will be used to establish the separation bounds when an externalf field ($\sigma\not\equiv0$) is given.

\begin{prop} \label{prop:well-separation.min.Riesz.(d-2).energy}
For $Q \equiv 0$ and $d \geq 2$ we have
\begin{equation} \label{eq:(d-2).separation.estimate}
\delta( \PSET_{n,d-2} ) \geq \frac{\kappa_{d}}{\left( n - 1 \right)^{1/d}} 
\end{equation}
for any $n(\geq 3)$-point Riesz $(d-2)$-energy \footnote{When $d = 2$ we mean logarithmic energy.} minimizing configuration $\PSET_{n,d-2}$ on $\mathbb{S}^d$, where
\begin{equation} \label{eq:(d-2).separation.estimate.constant}
\kappa_{d} = \left( \frac{1}{d} \, \frac{\gammafcn( ( d + 1 ) / 2 )}{\sqrt{\pi} \, \gammafcn( d / 2 )} \right)^{-1/d}.
\end{equation}
\end{prop}
Observe that $\kappa_d = ( 4 / W_{d-2}( \mathbb{S}^d ) )^{1/d}$ when $d \geq 3$.
The first three values of $\kappa_d$ are $\kappa_2 = 2$, $\kappa_3 = ( 3 \pi / 2 )^{1/3}$, and $\kappa_4 = 2 / 3^{1/4}$.
Curiously, $(\kappa_d)^d$ is the ratio of the volume of the unit ball in $\mathbb{R}^d$ divided by the surface area of the unit sphere in $\mathbb{R}^{d+1}$. (This constant also appears as the coefficient of the leading term in the asymptotic expansion of the $n$-point minimal Riesz $d$-energy as $n \to \infty$ (cf. \cite{KuSa1998}).)

Finally, we present a well-separation result for sequences of $(Q,s)$-Fekete sets in the hyper-singular case $s>d$. In this case the (strongly repellent) short-range interactions between points on the sphere ensure well-separation of minimizing configurations for any continuous external field on $\mathbb{S}^d$.
In fact, it is enough that $Q$ be integrable on some small subset of $\mathbb{S}^d$ of positive surface area measure. 

\begin{thm} \label{thm:main4}
Let $s > d$. Suppose there is a subset $B \subset \mathbb{S}^d$ such that $\sigma_d( B ) > 0$ and the fixed external field $Q$ is integrable over $B$ with respect to $\sigma_d$.
Then there is a constant $C$ independent of $n$ such that
\begin{equation} \label{eq:separation.estimate}
\delta( \PSET_{n,Q,s} ) \geq \frac{C}{n^{1/d}}
\end{equation}
for any $n$-point $(Q,s)$-Fekete set on $\mathbb{S}^d$.
\end{thm}

\begin{rmk}
In case of $Q = Q_n$ varies with $n$, there may be no single fixed subset $B$ satisfying the hypotheses in Theorem~\ref{thm:main4}. However, one  can still deduce well-separation by requiring the following: there is a sequence $\{ B_n \}$ of subsets of $\mathbb{S}^d$ such that for some $\eps > 0$, $\sigma_d( B_n ) \geq \eps$ for all $n$, and for some $C^\prime > 0$ independent of $n$,
\begin{equation*}
\frac{1}{\sigma_d( B_n )} \int_{B_n} \left| Q_n( \PT{x} ) \right| \, \dd \sigma_d( \PT{x} ) - \min\{ 0, \underline{M}^{Q_n} \} \leq C^\prime \, n^{s/d-1} \qquad \text{for all $n$,}
\end{equation*}
where $\underline{M}^{Q_n}$ denotes the minimum of $Q_n$ over $\mathbb{S}^d$. These conditions are derived from the main inequality \eqref{eq:master.inequality} and the estimate \eqref{eq:master.inequality.estimate}.
\end{rmk}

\begin{rmk}
For large classes of external fields $Q$ (e.g., continuous external fields), inequality \eqref{eq:master.inequality} and the estimate \eqref{eq:master.inequality.estimate} can be made explicit which, in turn, yields an explicit constant in the separation estimate~\eqref{eq:separation.estimate}, as the following example illustrates.
\end{rmk}

\begin{eg}
Let $s > d$. Consider the external field $Q( \PT{x} ) = q \, | \PT{x} - R \PT{p} |^{-s}$, $q \neq 0$, $R > 1$, due to a point source above the North Pole~$\PT{p}$. Clearly, $Q$ is continuous and thus integrable on $B = \mathbb{S}^d$. Thus Theorem~\ref{thm:main4} assures well-separation of $n$-point $(Q,s)$-Fekete sets on $\mathbb{S}^d$. An explicit lower bound can be easily derived from \eqref{eq:master.inequality} and \eqref{eq:master.inequality.estimate}. We find
\begin{equation*}
\left| \frac{1}{\sigma_d( B )} \int_D \big( Q( \PT{x} ) - \min\{ 0, Q( \PT{x}_j ) \} \big) \dd \sigma_d( \PT{x} ) \right| \leq |q| \int_{\mathbb{S}^d} \frac{\dd \sigma_d( \PT{x} )}{\left| \PT{x} - R \PT{p} \right|^{s}} - \min\Big\{ 0, \frac{q}{( R - 1 )^s} \Big\}
\end{equation*}
and, consequently,
\begin{equation*}
\delta( \PSET_{n,Q,s} ) \geq \left( \frac{\gamma_d}{1-\frac{1}{d} \gamma_d} \right)^{-1/s} \frac{g(n)}{n^{1/d}},
\end{equation*}
where
\begin{equation*}
g( n ) \DEF \Bigg\{ \frac{1}{s-d} + \frac{\beta_{s,d}}{2} \, n^{-2/d} + 2 n^{1-s/d} \left[ | q | \, U_s^{\sigma_d}( R \PT{p} ) - \min\Big\{ 0, \frac{q}{( R - 1 )^s} \Big\} \right] \Bigg\}^{-1/s}.
\end{equation*}
Note that $g(n) \to ( s - d )^{1/s}$ as $n \to \infty$. The constants $\gamma_d$ and $\beta_{s,d}$ are given in \eqref{eq:gamma.d} and \eqref{eq:coeff.beta.s.d},  respectively, and the representation of $U_s^{\sigma_d}$ appears in \eqref{eq:SignEqPot(b)}.
\end{eg}

\section{Negatively charged external fields}
\label{sec:neg.external.field}

In the following we consider external fields $Q$ that are generated by negative sources. The required lower semi-continuity of $Q : \mathbb{S}^d \to (-\infty,\infty]$ implies that no negative singularities can be on the sphere but it may support a negative ``continuous'' charge distribution (with no discrete part relative to $\mathbb{S}^d$).
We give a detailed analysis for the Riesz external field
\begin{equation} \label{eq:neg.external.field}
Q_{\PT{b},s}( \PT{x} ) = U_s^\sigma( \PT{x} ) = q \left| \PT{x} - \PT{b} \right|^{-s}, \quad \PT{x} \in \mathbb{S}^d, \qquad \text{$\PT{b} = - R \PT{p}$ ($R > 1$), $q < 0$,}
\end{equation}
where $\sigma = q \delta( \PT{b} )$, that is due to a negative point source at $\PT{b}$ below the South Pole and which also provides the basis for more general axis-supported fields defined by superposition of point source fields. Our analysis thus extends and complements results in \cite{BrDrSa2009} where positive axis-supported external fields were considered.

Intuitively, a negative point source under the South Pole will ``pull'' charge towards the South Pole and if sufficiently strong will cause a negatively charged spherical cap around the North Pole to appear on a grounded sphere. Grounding of the sphere imposes constant weighted potential everywhere on $\mathbb{S}^d$. This naturally leads to the signed equilibrium problem on the whole sphere or on its parts, say, the spherical cap $\Sigma_{t} \DEF \{ \PT{x} \in \mathbb{S}^d : \PT{x} \cdot \PT{p} \leq t \}$ centered at the South Pole. On a positively charged isolated sphere a sufficiently strong negative field will produce a spherical cap around the North Pole that is free of charge.
We are specifically interested in the charge distribution on the remaining part $\Sigma_{t_c}$, that is the $s$-extremal measure on $\mathbb{S}^d$ associated with the external field $Q_{\PT{b},s}$ and its support $\Sigma_{t_c}$.

We will use the methods and results of \cite{BrDrSa2009}. An essential concept is the \emph{$s$-balayage} of a measure. Recall that given a measure $\nu$ and a compact set $K$ (of the sphere $\mathbb{S}^d$), the balayage measure $\hat{\nu}:=\bal_s(\nu,K)$ preserves the Riesz $s$-potential of $\nu$ onto the set $K$ and diminishes it elsewhere (on the sphere $\mathbb{S}^d$). Let $\eta_t$ denote the signed $s$-equilibrium on $\Sigma_t$ associated with the external field $Q_{\PT{b},s}$. Then it can be expressed as
\begin{equation} \label{eq:eta}
\eta_t = \left[ \Phi_s(t) / W_s( \mathbb{S}^d) \right] \nu_t - q \epsilon_t,
\end{equation}
where
\begin{equation} \label{bal}
\epsilon_t = \epsilon_{t,s} \DEF \bal_s(\delta_{\PT{b}},\Sigma_t), \qquad \nu_t = \nu_{t,s} \DEF \bal_s(\sigma_d,\Sigma_t)
\end{equation}
are the $s$-balayage measures onto $\Sigma_t$ of the positive unit point charge at $\PT{b}$ and the uniform measure $\sigma_d$ on $\mathbb{S}^d$. The function $\Phi_s(t)$, defined by
\begin{equation} \label{eq:Phi}
\Phi_s(t) \DEF W_s(\mathbb{S}^d) \left( 1 + q \left\|\epsilon_t\right\| \right) \big/ \left\|\nu_t\right\|, \qquad d-2<s<d,
\end{equation}
in terms of the $s$-energy of $\mathbb{S}^d$, given in \eqref{eq:W.s.S.d} and norms $\|\epsilon_t\| = \int_{\mathbb{S}^d} \dd \epsilon_t$ and $\|\nu_t\| = \int_{\mathbb{S}^d} \dd \nu_t$, plays an important role in what follows. Indeed,
\begin{equation*}
\int_{\Sigma_t} \dd \eta_t = \frac{\Phi_s(t)}{W_s(\mathbb{S}^d)} \left\| \nu_t \right\| - q \left\| \epsilon_t \right\| = 1
\end{equation*}
and using that $U_s^{\nu_t}( \PT{x} ) = W_s( \mathbb{S}^d )$ and $U_s^{\epsilon_t}( \PT{x} ) = | \PT{x} - \PT{a} |^{-s}$ on $\Sigma_t$ by \eqref{bal}, at every $\PT{x} \in \Sigma_t$
\begin{equation} \label{eq:weighted.potential}
U_s^{\eta_t}( \PT{x} ) + Q_{\PT{b},s}( \PT{x} ) = \frac{\Phi_s(t)}{W_s(\mathbb{S}^d)} \, U_s^{\nu_t}( \PT{x} ) - q \, U_s^{\epsilon_t}( \PT{x} ) + Q_{\PT{b},s}( \PT{x} ) = \Phi_s(t).
\end{equation}
By Definition~\ref{def:signed.equilibrium}, $\sgnEqconst_{\Sigma_t,Q_{\PT{b},s},s} = \Phi_s(t)$ and \eqref{eq:functional.identity} relates $\Phi_s(t)$ to the $\mathcal{F}_s$-functional by means of $\Phi_s(t) = \mathcal{F}_s( \Sigma_t )$, whereas Proposition~\ref{prop:F.s.functional} implies that the latter is minimized by the support of the $s$-extremal measure on $\mathbb{S}^d$ associated with the external field \eqref{eq:neg.external.field} which turns out to be a spherical cap $\Sigma_{t_c}$. We will see that the unique minimum of $\Phi_s(t)$ in the interval $[-1,1]$ will provide this critical parameter $t_c$ (see Theorem~\ref{thm:s.equilibrium.measure}). Moreover, the remark following Theorem~\ref{thm:SignEq} provides the necessary and sufficient conditions (involving $\Phi_s(t)$ and therefore $\mathcal{F}_s( \Sigma_t )$) under which the signed $s$-equilibrium measure $\eta_t$ on $\Sigma_t$ turns into the $s$-extremal measure on $\mathbb{S}^d$ associated with the
external
field \eqref{eq:neg.external.field}.

Throughout, $\Hypergeom{2}{1}{a,b}{c}{z}$ and $\HypergeomReg{2}{1}{a,b}{c}{z}$ denote the Gauss hypergeometric function and its regularized form \footnote{The regularized form is well-defined even for $c$ a negative integer.} with series expansions
\begin{equation} \label{eq:HypergeomSeries}
\Hypergeom{2}{1}{a,b}{c}{z} \DEF \sum_{n=0}^\infty \frac{\Pochhsymb{a}{n}\Pochhsymb{b}{n}}{\Pochhsymb{c}{n}} \frac{z^n}{n!}, \quad  \HypergeomReg{2}{1}{a,b}{c}{z} \DEF \sum_{n=0}^\infty \frac{\Pochhsymb{a}{n}\Pochhsymb{b}{n}}{\gammafcn(n+c)} \frac{z^n}{n!}, \qquad  |z|<1,
\end{equation}
where $\Pochhsymb{a}{0} \DEF 1$ and $\Pochhsymb{a}{n} \DEF a (a+1) \cdots (a+n-1)$ for $n\geq1$ is the Pochhammer symbol. We also recall that the incomplete Beta
function and the Beta function are defined as
\begin{equation} \label{eq:betafnc}
\betafcn(x;\alpha,\beta) \DEF \int_{0}^x v^{\alpha-1} \left( 1 - v \right)^{\beta-1} \dd v, \qquad \betafcn(\alpha,\beta) \DEF \betafcn(1; \alpha,\beta),
\end{equation}
whereas the regularized incomplete Beta function is given by
\begin{equation}
\mathrm{I}(x;a,b) \DEF \betafcn(x;a,b) \big/ \betafcn(a,b). \label{regbetafnc}
\end{equation}

First, we give the representation of the signed equilibrium on the whole sphere $\mathbb{S}^d$, which is well-known from elementary physics (cf. \cite[p.~61]{Ja1998}) in the classical Coulomb case, and provide a necessary and sufficient condition when it also is the $s$-extremal measure on $\mathbb{S}^d$. 

\begin{prop} \label{prop:SignEq} 
Let $0 < s < d$ and $R>1$. The signed $s$-equilibrium $\eta_{\PT{b}} = \eta_{\mathbb{S}^d, Q_{\PT{b},s},s}$ on $\mathbb{S}^d$ associated with the Riesz external field \eqref{eq:neg.external.field}, where in fact $q \in \mathbb{R} \setminus \{0\}$, is given by
\begin{equation} \label{eq:signedeqdens}
\dd \eta_{\PT{b}}(\PT{x}) = \eta_{\PT{b}}^\prime(\PT{x}) \dd \sigma_d(\PT{x}), \quad \eta_{\PT{b}}^\prime(\PT{x}) \DEF 1 + \frac{q U_s^{\sigma_d}(\PT{b})}{W_s(\mathbb{S}^d)} - \frac{q\left(R^2-1\right)^{d-s}}{W_s(\mathbb{S}^d) \left| \PT{x} - \PT{b} \right|^{2d-s}}.
\end{equation}
Furthermore,
\begin{equation*}
U_s^{\eta_{\PT{b}}}(\PT{z}) + Q_{\PT{b},s}(\PT{z}) = \mathcal{F}_s( \mathbb{S}^d ) = \Phi_s( 1 ) = W_s( \mathbb{S}^d ) + q U_s^{\sigma_d}( \PT{b} ), \qquad \PT{z} \in \mathbb{S}^d,
\end{equation*}
where $U_s^{\sigma_d}(\PT{b}) = \int_{\mathbb{S}^d} k_s (\PT{b},\PT{y})\, \dd \sigma_d(\PT{y})$ has the following representation:
\begin{equation} \label{eq:SignEqPot(b)}
U_s^{\sigma_d} (\PT{b}) = \left( R + 1 \right)^{-s} \Hypergeom{2}{1}{s/2,d/2}{d}{4R \big/ \left(R+1\right)^2}.
\end{equation}

Moreover, if $q < 0$, then $\supp(\mu_{Q_{\PT{b},s}}) = \mathbb{S}^d$ if and only if
\begin{equation} \label{eq:Gonchar.condition}
\frac{W_s(\mathbb{S}^d)}{q} \leq \frac{\left(R-1\right)^{d-s}}{\left(R+1\right)^{d}} - U_s^{\sigma_d}(\PT{b}).
\end{equation}
In such a case $\mu_{Q_{\PT{b},s}} = \eta_{\PT{b}}$.
\end{prop}

\begin{proof}
Indeed, it can be readily verified that $\eta_{\PT{b}}$ given by \eqref{eq:signedeqdens} is the signed equilibrium measure on the whole sphere on observing that $\epsilon_{\PT{b}} \DEF \bal_s( \delta_{\PT{b}}, \mathbb{S}^d )$ with
\begin{equation*}
\dd \epsilon_{\PT{b}}( \PT{x} ) = \dd \epsilon_1( \PT{x} ) = \frac{\left(R^2-1\right)^{d-s}}{W_s(\mathbb{S}^d) \left| \PT{x} - \PT{b} \right|^{2d-s}} \dd \sigma_d( \PT{x} ), \qquad \left\| \epsilon_{\PT{b}} \right\| = \frac{1}{W_s( \mathbb{S}^d )} \, U_s^{\sigma_d}( \PT{b} ),
\end{equation*}
by using a suitably defined Kelvin transformation; for details see \cite[Proof of Theorem~2]{BrDrSa2009} which can be easily extended to hold for the external field \eqref{eq:neg.external.field}.

The unique signed equilibrium $\eta_{\PT{b}}$ on $\mathbb{S}^d$ associated with $Q_{\PT{b},s}$ coincides with the unique positive $s$-equilibrium measure on $\mathbb{S}^d$ if and only if $\eta_{\PT{b}}$ is a measure. This is a consequence of the uniqueness of either measure and variational inequalities (Proposition~\ref{prop:1}(d)). As the strictly decreasing density function $\eta_{\PT{b}}^\prime$ assumes its minimum value at the North Pole~$\PT{p}$,
\begin{equation*}
\eta_{\PT{b}}^\prime( \PT{p} ) = 1 + \frac{q U_s^{\sigma_d} (\PT{b})}{W_s(\mathbb{S}^d)} - \frac{q}{W_s(\mathbb{S}^d)} \, \frac{\left(R-1\right)^{d-s}}{\left(R+1\right)^{d}},
\end{equation*}
we obtain the necessary and sufficient criterion in \eqref{eq:Gonchar.condition}.
\end{proof}

In light of \eqref{eq:Gonchar.condition} it is natural to ask if there is a critical distance $R_q > 1$ such that the support of $\mu_{Q_{\PT{b},s}}$ is all of $\mathbb{S}^d$ for $R \geq R_q$ but is a proper subset of $\mathbb{S}^d$ for $R < R_q$. For external fields with positive charge $q$ such a critical distance always exists; cf. \cite[Remark~4 and Example~5]{BrDrSa2009} and \cite{BrDrSa2012}. However, for negative charge $q$ such a critical distance does not always exist. We discuss this phenomenon in more detail in \cite{BrDrSaXXXXpreparation}.

How are the Riesz-$s$ external fields, generated by a positive point charge $q_+$ above the North Pole ($\PT{a} = R_+ \PT{p}$) and a negative point charge $q_-$ below the South Pole ($\PT{b} = - R_- \PT{p}$), related? Let $\eta_{\PT{a}}^\prime$ and $\eta_{\PT{b}}^\prime$ be the density functions of the respective signed equilibrium on $\mathbb{S}^d$. Then it can be readily verified that for same center distances $| \PT{a} | = | \PT{b} | = R > 1$ and total charge $q_-+q_+=0$ the densities complement each other; that is,
\begin{equation*}
\eta_{\PT{a}}^\prime( \PT{x} ) + \eta_{\PT{b}}^\prime( \PT{y} ) = 2 \qquad \text{for all $\PT{x}, \PT{y} \in \mathbb{S}^d$ with $\PT{p} \cdot \PT{x} + \PT{p} \cdot \PT{y} = 0$.}
\end{equation*}
On the other hand, when it is assumed that both $\eta_{\PT{a}}^\prime$ and $\eta_{\PT{b}}^\prime$ attain the same value at the North Pole, then again by \cite[Theorem~2]{BrDrSa2009} and Proposition~\ref{prop:SignEq},
\begin{equation} \label{eq:balance.relation}
1 + \frac{q_- U_s^{\sigma_d}(\PT{b})}{W_s(\mathbb{S}^d)} - \frac{q_- \left(R_-^2-1\right)^{d-s}}{W_s(\mathbb{S}^d) \left( R_- + 1 \right)^{2d-s}} = 1 + \frac{q_+ U_s^{\sigma_d}(\PT{a})}{W_s(\mathbb{S}^d)} - \frac{q_+ \left(R_+^2-1\right)^{d-s}}{W_s(\mathbb{S}^d) \left( R_+ - 1 \right)^{2d-s}}.
\end{equation}
Here, in the case $| \PT{a} | = | \PT{b} | = R > 1$ (i.e. $R_- = R_+ = R$), Equation~\eqref{eq:balance.relation} is equivalent with
\begin{equation} \label{eq:rel.q+.q-}
\left( R^2 - 1 \right)^{s/2} U_s^{\sigma_d}( \PT{a} )
= \left( \frac{R+1}{R-1} \right)^{d-s/2} \frac{q_+}{q_+ - q_-} + \left( \frac{R-1}{R+1} \right)^{d-s/2} \frac{-q_-}{q_+ - q_-}
\end{equation}
and in the case $q_- + q_+ = 0$ (i.e. $q_+ = - q_- = q > 0$), Equation~\eqref{eq:balance.relation} becomes
\begin{equation} \label{eq:rel.R-.R+}
- U_s^{\sigma_d}(\PT{b}) + \frac{\left(R_- - 1\right)^{d-s}}{\left( R_- + 1 \right)^{d}} = U_s^{\sigma_d}(\PT{a}) - \frac{\left(R_+ + 1\right)^{d-s}}{\left( R_+ - 1 \right)^{d}}.
\end{equation}
Observe that the last relation between $R_-$ and $R_+$ does not depend on the charge~$q$. Figure~\ref{fig:signed.equilibria.sphere} illustrates theses cases.

\begin{figure}[ht]
\includegraphics[scale=.825]{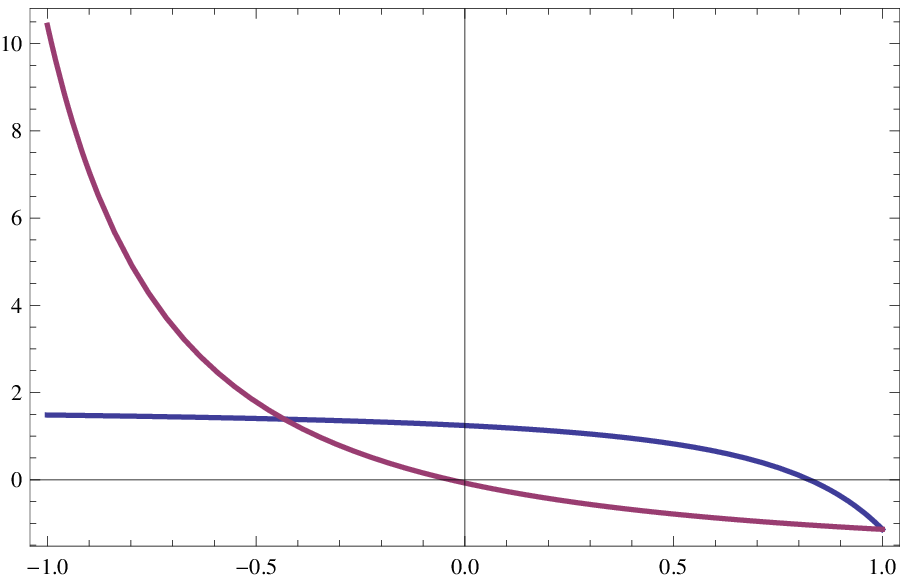} \phantom{\includegraphics[scale=.825]{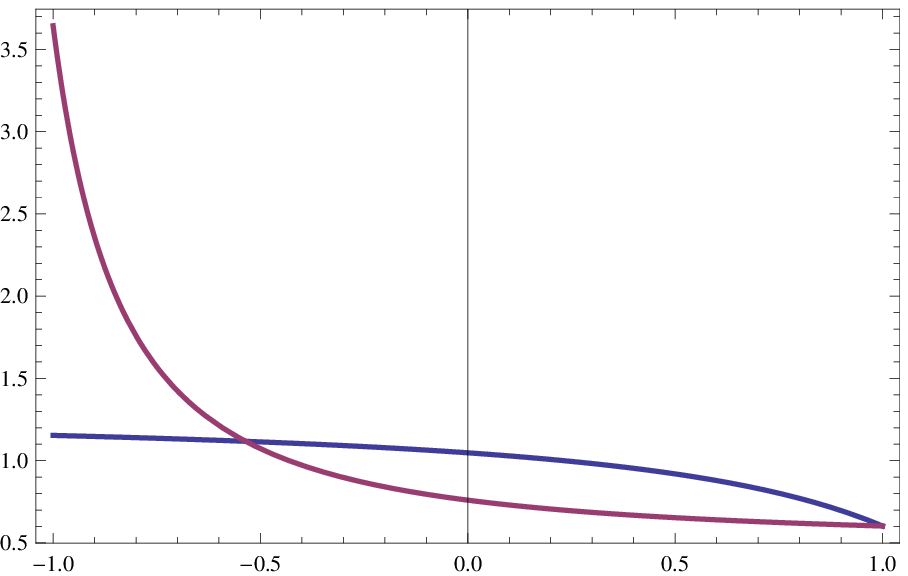}} \\[1.5mm]
\includegraphics[scale=.825]{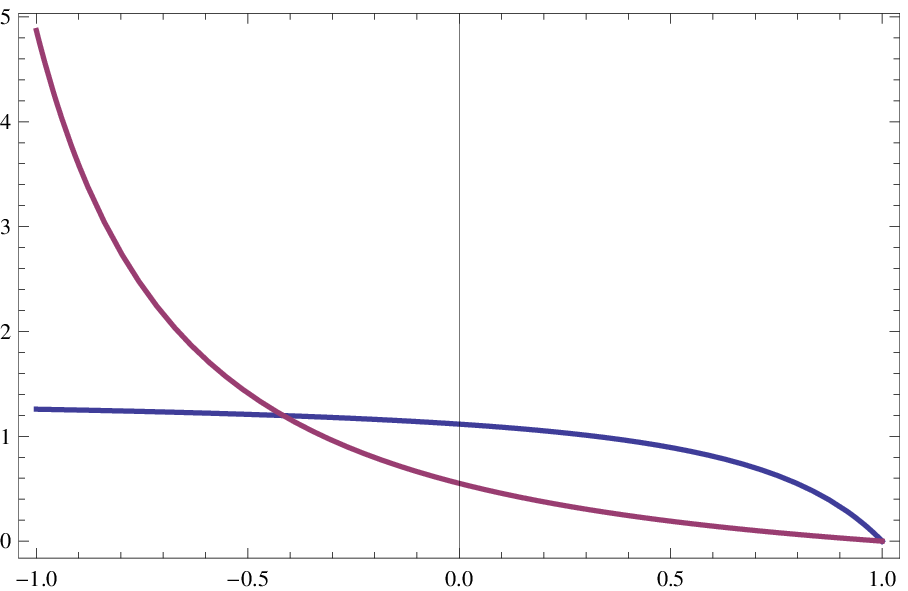} \includegraphics[scale=.825]{SignedEquSphereQ0s1.eps} \\[1.5mm]
\includegraphics[scale=.825]{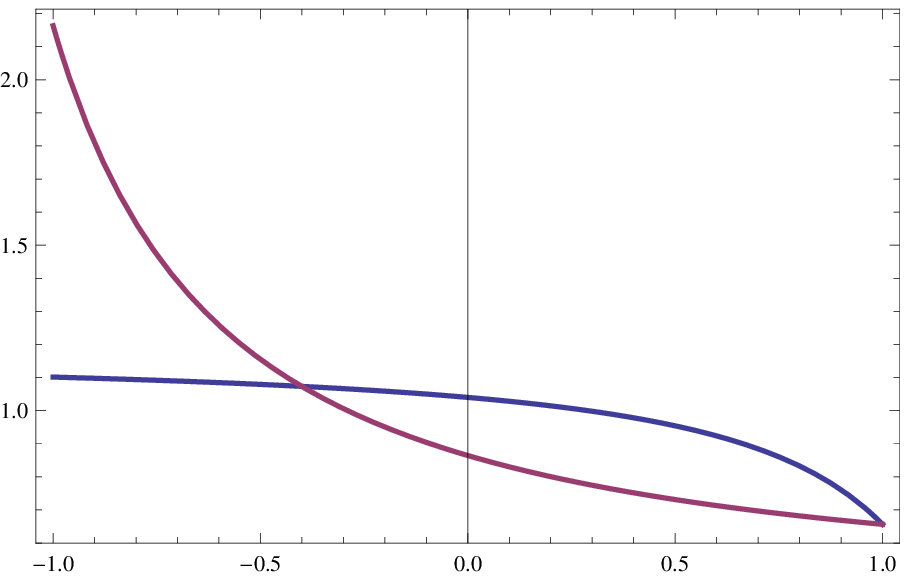} \includegraphics[scale=.825]{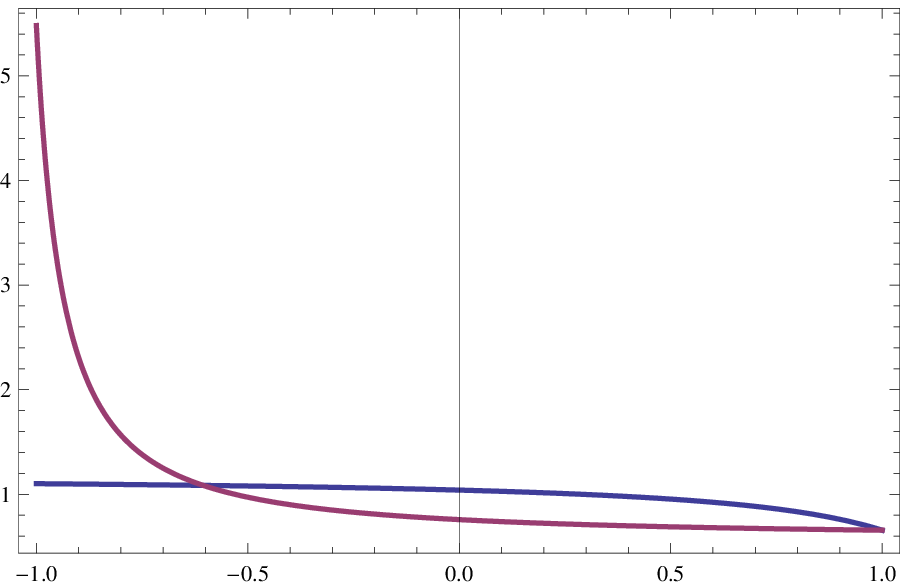}
\caption{\label{fig:signed.equilibria.sphere} Comparison of density functions of signed equilibria on $\mathbb{S}^2$ associated with Riesz external fields generated by a positive point charge above the North Pole (Blue curves in colored version, i.e. smaller value at~$-1$) and a negative point charge below the South Pole (Red curves in colored version, i.e. larger value at~$-1$) having the same density function value at $1$ (North Pole). Left column: fixed center distance $R = 1 + \phi$ ($\phi$ Golden ratio), $q_+ = 1$ and $q_-$ from \eqref{eq:rel.q+.q-}. Right column: total charge $q_- + q_+ = 0$ with $q_+ = -q_- = q = 1$, $R_+ = 1 + \phi$ and $R_-$ from \eqref{eq:rel.R-.R+}. From top to bottom: $s= 1/2$, $1$, $3/2$. (No solution for $s = 1/2$ in right column.)}
\end{figure}

Next, we investigate the signed equilibrium on spherical caps $\Sigma_t$. Ultimately, our goal is to use this signed equilibrium to obtain the $s$-extremal measure on $\mathbb{S}^d$ associated with the external field in \eqref{eq:neg.external.field} when its support is not all of $\mathbb{S}^d$. We remark that the proofs of the results and remarks in this section can be obtained by inspecting the proofs of the related results in \cite{BrDrSa2009} (details will be presented in a later paper \cite{BrDrSaXXXXpreparation}).

\begin{thm} \label{thm:SignEq} Let $d-2<s<d$. The signed $s$-equilibrium $\eta_t$
on the spherical cap $\Sigma_t \subset \mathbb{S}^d$, $-1 < t < 1$, associated with $Q_{\PT{b},s}$ in \eqref{eq:neg.external.field} is
given by \eqref{eq:eta}. It is absolutely continuous in the sense that for $\PT{x} = ( \sqrt{1-u^2} \, \overline{\PT{x}}, u) \in \Sigma_t$,
\begin{equation} \label{eq:eta.t}
\dd \eta_{t}(\PT{x}) = \eta_{t}^{\prime}(u) \frac{\omega_{d-1}}{\omega_{d}} \left( 1 - u^2 \right)^{d/2-1} \dd u \dd\sigma_{d-1}(\overline{\PT{x}}),
\end{equation}
where (with $R=|\PT{b}|$ and $r = \sqrt{R^2 + 2 R t + 1}$)
\begin{equation}
\begin{split} \label{eta.t.prime.1st.result}
\eta_{t}^{\prime}(u) &= \frac{1}{W_s(\mathbb{S}^d)} \frac{\gammafcn(d/2)}{\gammafcn(d-s/2)}
\left( \frac{1-t}{1-u} \right)^{d/2} \left( \frac{t-u}{1-t} \right)^{(s-d)/2} \\
&\phantom{=\times}\times \Bigg\{ \Phi_s (t)
\HypergeomReg{2}{1}{1,d/2}{1-(d-s)/2}{\frac{t-u}{1-u}}  \\
&\phantom{=\times\pm}-  \frac{q\left( R - 1 \right)^{d-s}}{r^{d}}
\HypergeomReg{2}{1}{1,d/2}{1-(d-s)/2}{\frac{\left(R+1\right)^{2}}{r^{2}}
\, \frac{t-u}{1-u}}  \Bigg\}.
\end{split}
\end{equation}
The density $\eta_{t}^{\prime}$ is expressed in terms of regularized Gauss hypergeometric functions.

Furthermore, if $\PT{z} = ( \sqrt{1-\xi^2}\; \overline{\PT{z}},
\xi)\in \mathbb{S}^d$, the weighted $s$-potential is given by
\begin{align}
U_s^{\eta_t}(\PT{z})+Q_{\PT{b},s}(\PT{z}) &= \Phi_s(t), \qquad \PT{z} \in \Sigma_t, \label{eq:weighted.inside} \\
\begin{split}
U_s^{\eta_t}(\PT{z})+Q_{\PT{b},s}(\PT{z}) &= \Phi_s(t) + \frac{q}{\rho^s} \, \mathrm{I}\Big(\frac{(R-1)^2}{r^2} \frac{\xi-t}{1+\xi};
\frac{d-s}{2}, \frac{s}{2} \Big) \\
&\phantom{=\pm}- \Phi_s(t) \, \mathrm{I}\Big(\frac{\xi-t}{1+\xi};
\frac{d-s}{2}, \frac{s}{2}\Big), \qquad \PT{z} \in \mathbb{S}^d \setminus \Sigma_t, \label{eq:weighted.outside}
\end{split}
\end{align}
where $\rho=\sqrt{R^2+2R\xi+1}$ and $\mathrm{I}(x;a,b)$ is the regularized incomplete Beta function. 
\end{thm}

\begin{rmk}
There is a simple relation between the positive and negative point charge problem with regard to their signed equilibria when these charges are on opposite sides of an axis. Namely, if $Q^+( \PT{x} ) \DEF q_+ / | \PT{x} - \PT{a} |^s$ and $Q^-( \PT{y} ) \DEF q_- / | \PT{y} - \PT{b} |^s$ are the Riesz external fields generated by a positive point charge at $\PT{a}$ above the North Pole and a negative point charge at $\PT{b}$ below the South Pole, then the weighted $s$-potentials of the respective signed equilibria $\eta_{\PT{a}}$ and $\eta_{\PT{b}}$ on the spherical cap $\Sigma_t$ satisfy
\begin{equation*}
U_s^{\eta_{\PT{a}}}( \PT{x} ) + Q^+( \PT{x} ) = \sgnEqconst_{\Sigma_t, Q^+, s}, \qquad U_s^{\eta_{\PT{b}}}( \PT{y} ) + Q^-( \PT{y} ) = \sgnEqconst_{\Sigma_t, Q^-, s}
\end{equation*}
everywhere on $\mathbb{S}^d$.
Furthermore, if $q_-+q_+ = 0$, then the following principle holds:
\begin{equation*}
U_s^{\eta_{\PT{a}}}( \PT{x} ) + U_s^{\eta_{\PT{b}}}( \PT{y} ) = \sgnEqconst_{\Sigma_t, Q^+, s} + \sgnEqconst_{\Sigma_t, Q^-, s}
\end{equation*}
for all $\PT{x}, \PT{y} \in \mathbb{S}^d$ such that $| \PT{x} - \PT{a} | = | \PT{y} - \PT{b} |$.
\end{rmk}

The next remark, leading up Theorem~\ref{thm:s.equilibrium.measure}, emphasizes the special role of $\Phi_s( t )$.

\begin{rmk}
It can be shown that the signed equilibrium $\eta_t$ on $\Sigma_t$ associated with $Q_{\PT{b},s}$ is a positive measure with support $\Sigma_t$ if and only if
\begin{equation} \label{eq:weighted.neccessary.density}
\Phi_s( t ) \geq q \frac{\left( R - 1 \right)^{d-s}}{\left( R^2 + 2 R t + 1 \right)^{d/2}},
\end{equation}
whereas the weighted $s$-potential of the signed equilibrium $\eta_t$ on $\Sigma_t$ associated with $Q_{\PT{b},s}$ exceeds the value $\Phi_s(t)$ assumed on $\Sigma_t$ strictly \emph{outside} of $\Sigma_t$ (but on $\mathbb{S}^d$) if and only if
\begin{equation} \label{eq:weighted.neccessary}
W_s( \mathbb{S}^d ) \, \frac{1 + q \left\| \epsilon_t \right\|}{\left\| \nu_t \right\|} = \Phi_s( t ) \leq q \frac{\left( R - 1 \right)^{d-s}}{\left( R^2 + 2 R t + 1 \right)^{d/2}}.
\end{equation}
As the right-hand side above is negative, one also has the weaker restriction $\| \epsilon_t \| > -1 / q$.
\end{rmk}

For $\eta_t$ to coincide with the $s$-extremal measure on $\mathbb{S}^d$ associated with $Q_{\PT{b},s}$ with support $\Sigma_t$ both \eqref{eq:weighted.neccessary} and \eqref{eq:weighted.neccessary.density} have to hold (cf. Proposition~\ref{prop:1}(d)). The difficult part of the next statement is to verify that the arising equation has a unique solution, which can be done as in the proof of \cite[Theorem~13]{BrDrSa2009}.

\begin{thm} \label{thm:s.equilibrium.measure}
Let $d - 2 < s < d$. For the external field \eqref{eq:neg.external.field} the function $\Phi_s (t)$ given in \eqref{eq:Phi} has precisely one global minimum $t_c\in (-1,1]$. This minimum is either the unique solution $t_c\in (-1,1)$ of the equation
\begin{equation*}
\Phi_s (t) = q \left( R - 1 \right)^{d-s} \big/ \left( R^2 + 2 R t + 1 \right)^{d/2},
\end{equation*}
or $t_c=1$ when such a solution does not exist. In addition, $\Phi_s(t)$ is greater than the right-hand side above if $t \in (-1,t_c)$ and is less than if $t \in (t_c,1)$.  Moreover, $t_c = \max \{ t : \eta_t \geq 0 \}$. The extremal measure $\mu_{Q_{\PT{b},s}}$ on $\mathbb{S}^d$ is given by $\eta_{t_c}$ (see \eqref{eq:eta.t}), and $\supp( \mu_{Q_{\PT{b},s}} ) = \Sigma_{t_c}$. 
\end{thm}

In the limiting case $s = d - 2$ with $s > 0$ it can be shown that the $s$-balayage measures
\begin{equation} \label{barBal}
\overline{\epsilon}_t \DEF \epsilon_{t,d-2} = \bal_{d-2}(\delta_{\PT{b}},\Sigma_t), \qquad  \overline{\nu}_t \DEF \nu_{t,d-2} = \bal_{d-2}(\sigma,\Sigma_t)
\end{equation}
exist and both have a component that is uniformly distributed on the boundary of $\Sigma_t$.
Each of these measures is the weak-star limit as $s \to (d-2)^+$ of the respective measure in \eqref{bal}.
An inspection of the proofs in \cite{BrDrSa2009} yields that the signed $s$-equilibrium $\overline{\eta}_t$ on the spherical cap $\Sigma_t$ associated with $\overline{Q}_{\PT{b},d-2}(\PT{x}) = q \, | \PT{x} - \PT{b} |^{2-d}$ ($q < 0$ and $\PT{b} = ( \PT{0}, -R )$ with $R > 1$) is given by
\begin{equation*}
\overline{\eta}_t = \left[ \overline{\Phi}_{d-2}(t) / W_{d-2}(\mathbb{S}^d) \right] \overline{\nu}_t - q \overline{\epsilon}_t, \qquad \overline{\Phi}_{d-2}(t) \DEF  W_{d-2}(\mathbb{S}^d) \left( 1 + q \left\|\overline{\epsilon}_t\right\| \right) / \left\|\overline{\nu}_t\right\|.
\end{equation*}
More explicitly, it can be decomposed into a continuous and a discrete part by means of
\begin{equation} \label{etabar}
\dd \overline{\eta}_t(\PT{x}) = \overline{\eta}_t^\prime(u) \, \dd \sigma_{d}\big|_{\Sigma_t}(\PT{x}) + \overline{q}_t \, \delta_t (u) \dd \sigma_{d-1}(\overline{\PT{x}}), \qquad \PT{x} = ( \sqrt{1 - u^2} \, \overline{\PT{x}}, u ) \in \mathbb{S}^d,
\end{equation}
where the density with respect to $\sigma_d$ restricted to $\Sigma_t$ has the form
\begin{equation*}
\overline{\eta}_t^\prime(u) = \frac{\overline{\Phi}_{d-2}(t)}{W_{d-2}(\mathbb{S}^d)} - \frac{q}{W_{d-2}(\mathbb{S}^d)} \, \frac{\left( R^2 - 1 \right)^2}{\left( R^2 + 2 R u + 1 \right)^{d/2+1}}
\end{equation*}
and the boundary charge uniformly distributed over the boundary of $\Sigma_t$ is
\begin{equation*}
\overline{q}_t = \frac{1-t}{2} \left( 1 - t^2 \right)^{d/2-1} \left[ \overline{\Phi}_{d-2}(t) - \frac{q \left( R - 1 \right)^2}{\left( R^2 + 2 R t + 1 \right)^{d/2}} \right].
\end{equation*}
The vanishing of this boundary charge characterizes the critical distance $t_c$ for which $\overline{\eta}_{t_c}$ becomes the $(d-2)$-extremal measure on $\mathbb{S}^d$ associated with $\overline{Q}_{\PT{b},d-2}$ (details will be presented in a later paper \cite{BrDrSaXXXXpreparation}).
A similar result holds for the logarithmic case $s = \log$ on $\mathbb{S}^2$.

\section{Examples of external field problems for small $N$}
\label{sec:Examples}

\subsection{The three point problem with a convex external field}
Let $\PT{a}$ be a fixed point in $\mathbb{R}^3$ with $R=| \PT{a} |>1$. It will be regarded as the source of an external field $Q(\PT{x})=f(|\PT{x}-\PT{a}|^2)$, where $f$ is a strictly convex and decreasing function. Then as we show below the optimal configuration (minimizing \eqref{eq:discrete.Riesz.energy}) on $\mathbb{S}^2$ for every $s>0$ is an equilateral triangle perpendicular to the main axis passing through $\PT{a}$ and the center of the sphere; cf. Figure~\ref{fig3}.
The intercept $t_0$ of the plane supporting this unique (up to rotation about the main axis) triangle with the main axis varies with~$f$. From these facts it is easy to see that
\begin{equation*}
\mathcal{E}_s^Q(3) = \frac{6}{\left[ 3 \left( 1 - t_0^2 \right) \right]^{s/2}} + 12 f( 1 - 2 R t_0 + R^2 ),
\end{equation*}
where the negative intercept $t_0$ is the unique minimum of $\mathcal{E}_s^Q(3)$ satisfying the relation
\begin{equation*}
\frac{3 s \left( - t_0 \right)}{\left[ 3 \left( 1 - t_0^2 \right) \right]^{s/2+1}} = 4 R \left[ - f^\prime( 1 - 2 R t_0 + R^2 ) \right].
\end{equation*}

\begin{rmk}
It is interesting to note that for the case $Q(\PT{x}) = \frac{1}{2} | \PT{x} - \PT{p} |^{-s}$ we deduce the well-known fact that the tetrahedron has minimal Riesz-$s$ energy for four points on $\mathbb{S}^2$ interacting via the Riesz-$s$ potential; see the first frame of Figure~\ref{fig3}.
\end{rmk}

\begin{figure}[h]
\includegraphics[scale=.4]{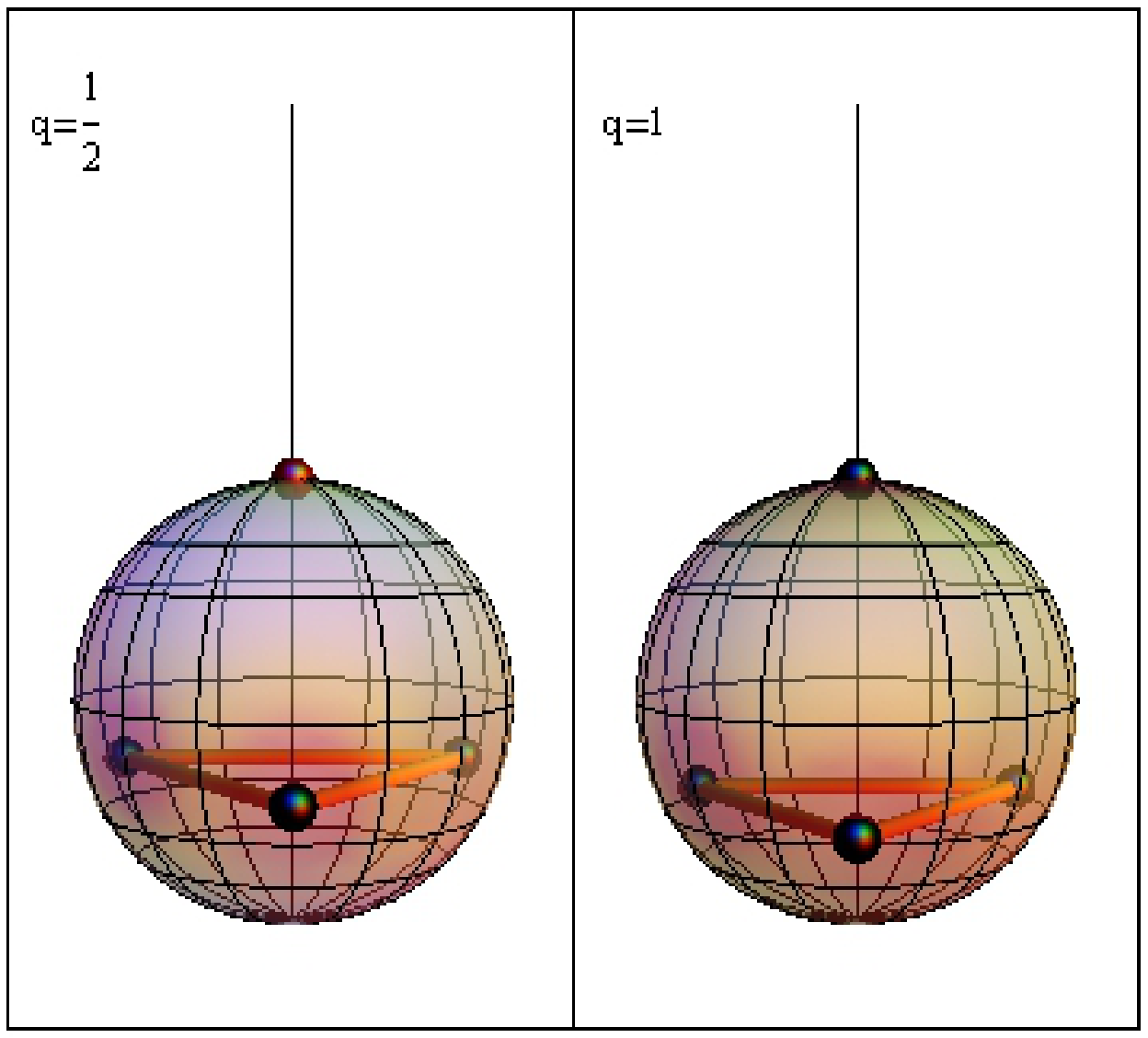}
\includegraphics[scale=.4]{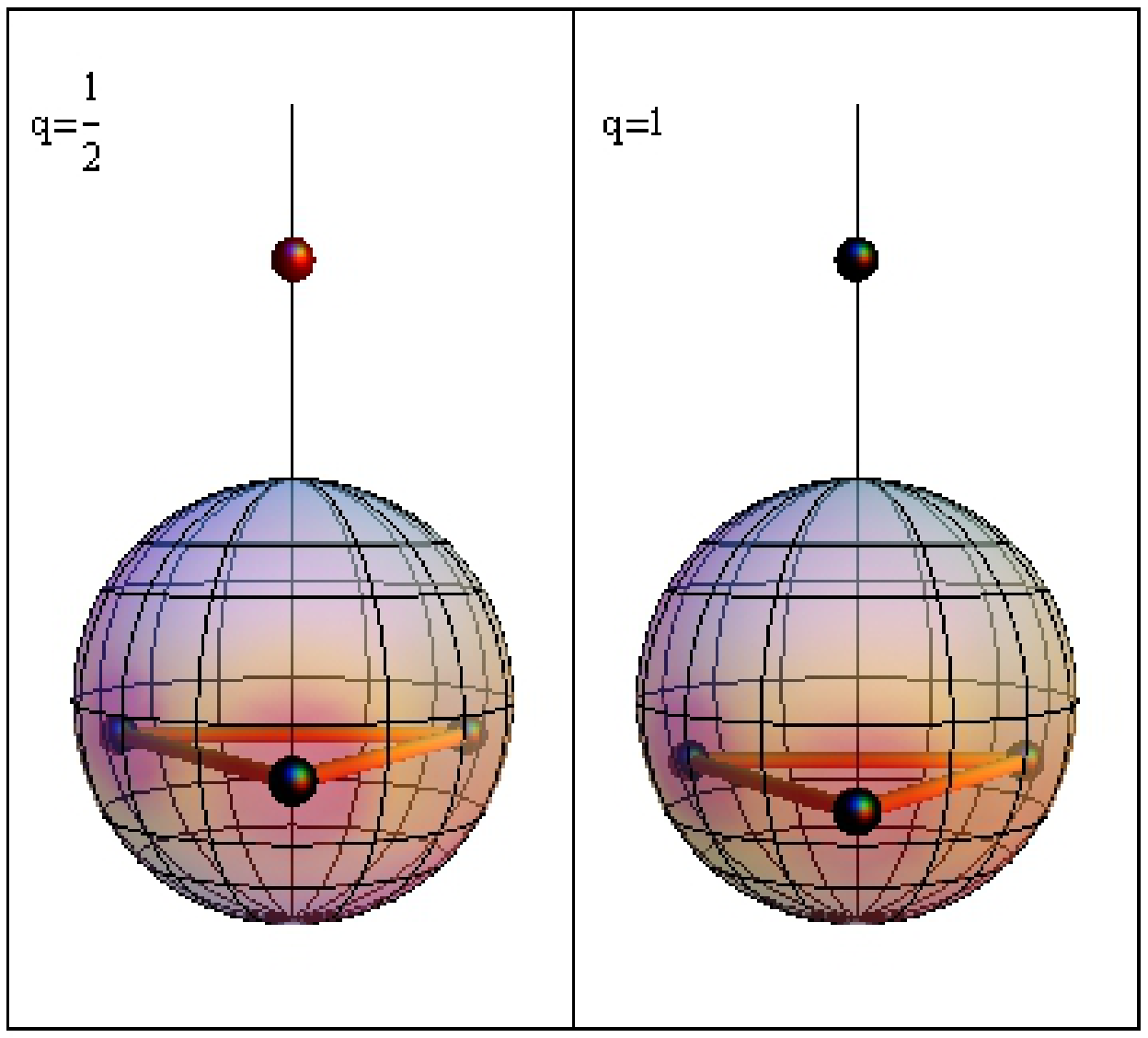}
\includegraphics[scale=.4]{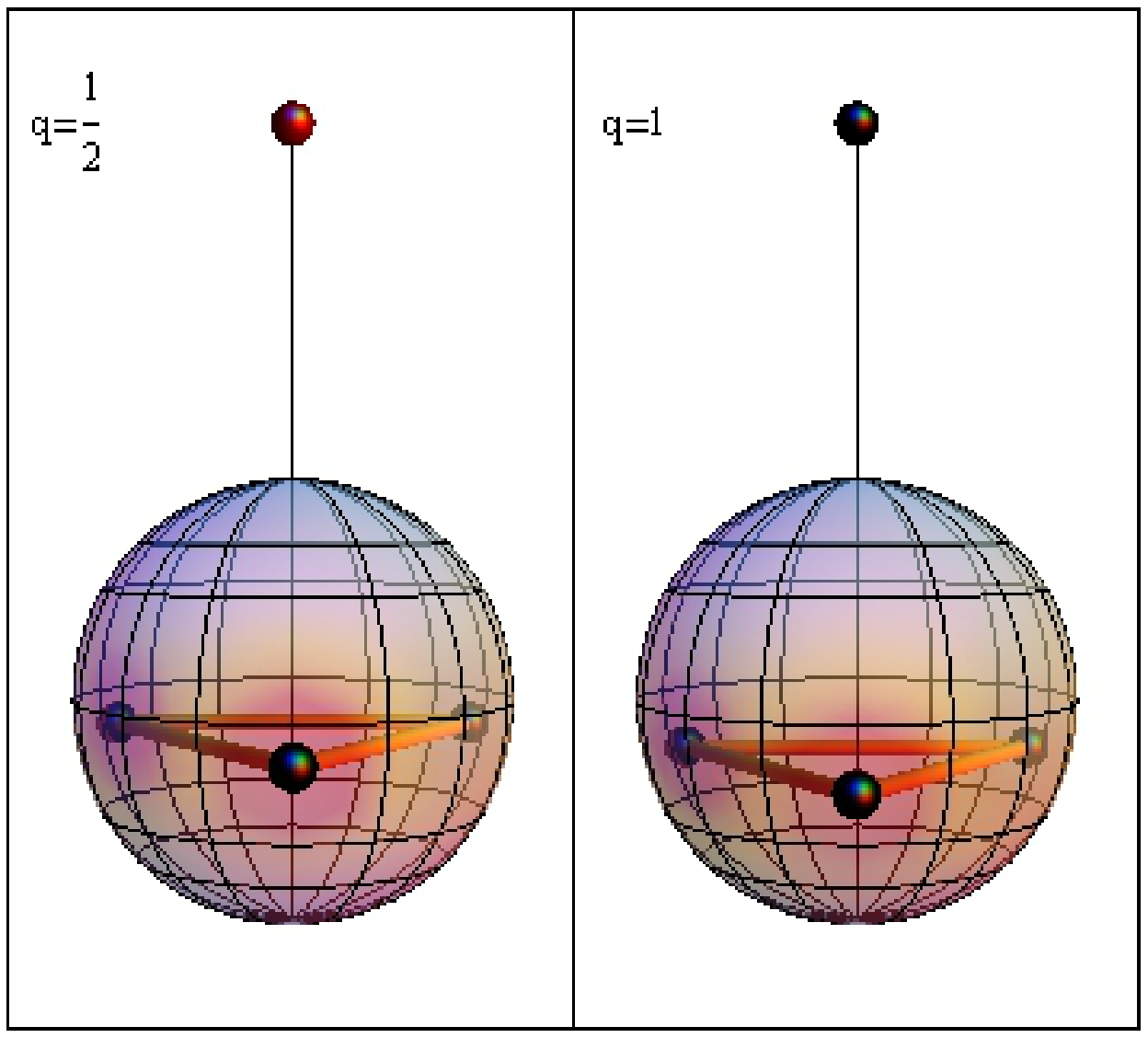}
\caption{\label{fig3} Typical $3$-point $(Q,s)$-Fekete sets for weaker (${q = 1/2}$, left) and stronger (${q = 1}$, right) Riesz external fields $Q( \PT{x} ) = q | \PT{x} - R \PT{p} |^{-s}$ in the Coulomb case $s = 1$ for selected values of center distances $R = 1$, $2$, and $1 + \phi$ ($\phi$ is the Golden ratio); compare with Fig.~\ref{fig2}.}
\end{figure}

We now establish that the aforementioned configuration is indeed optimal. For this purpose let $\{ \PT{x}_1,\PT{x}_2,\PT{x}_3\} $ be an optimal
configuration. Without loss of generality, we assume that the segment connecting $\PT{x}_1$ and $\PT{x}_2$ lies in the $yz$-plane and is parallel to the
$y$-axis, say, with parameterization $\PT{x}_1=(0,\sqrt{1-h^2},h)$ and $\PT{x}_2=(0,-\sqrt{1-h^2},h)$, $-1 < h < 1$. Furthermore, we can write
\begin{equation*}
\PT{x}_3=(\sqrt{1-t^2}\cos \alpha,\sqrt{1-t^2}\sin \alpha,t) \quad \text{and} \quad \PT{a}=(R \sin \theta \cos \beta, R \sin \theta \sin \beta, R \cos \theta).
\end{equation*}
Since $|\PT{x}_3-\PT{x_1}|^2+|\PT{x}_3-\PT{x_2}|^2 = 2(2-2ht)$, the
convexity of the function $u \mapsto u^{-s/2}$ implies that
\begin{equation*}
\left(|\PT{x}_3-\PT{x_1}|^2\right)^{-s/2}+\left(|\PT{x}_3-\PT{x_2}|^2\right)^{-s/2} \geq 2 (2-2ht)^{-s/2},
\end{equation*}
where equality holds only if $\alpha=0$ or $\pi$.
Similarly, since
\begin{equation*}
\left|\PT{a}-\PT{x_1}\right|^2+\left|\PT{a}-\PT{x_2}\right|^2 = 2\left(R^2+1-2hR\cos \theta\right),
\end{equation*}
the convexity of $f$ implies that
\begin{equation*}
Q(\PT{x_1})+ Q(\PT{x_2})=f(|\PT{a}-\PT{x_1}|^2)+f(|\PT{a}-\PT{x_2}|^2) \geq 2f(R^2+1-2hR\cos \theta),
\end{equation*}
where equality holds only if $\beta=0$ or $\pi$.
Moreover, from the fact that $f$ is strictly decreasing we get
\begin{align*}
Q(\PT{x}_3) = f(|\PT{a}-\PT{x_3}|^2)
&= f(R^2+1-2R\sqrt{1-t^2} \sin \theta \cos (\alpha-\beta) +2Rt\cos\theta) \\
&\geq f(R^2+1+2R\sqrt{1-t^2}\sin \theta +2Rt\cos \theta),
\end{align*}
and equality holds only if $\alpha-\beta = \pi$ or $-\pi$.
Consequently, the optimal configuration must satisfy
$|\PT{x}_3-\PT{x}_1|=|\PT{x}_3-\PT{x}_2|$ and
$|\PT{a}-\PT{x}_1|=|\PT{a}-\PT{x}_2|$. Since the choice of
$\PT{x}_1$ and $\PT{x}_2$ was arbitrary, the same is true for any
pair, which implies that the configuration is an equilateral
triangle with equidistant points from $\PT{a}$. Finally, it is easy to see
that there is a unique solution for the intercept $t_0$ of the
equilateral triangle and the axes.

\subsection{The four point problem with Riesz external field}
\label{sec:4.point.problem}

Let four unit charges be restricted to move on the unit sphere
$\mathbb{S}^2$ and the fifth point (placed on the polar axis above
the North Pole) act as the source of an external field
$Q_{\PT{a},s}( \PT{x} ) = q k_s( \PT{x}, \PT{a} )$, where $\PT{a} =
R \PT{p}$, $R \geq 1$, and $\PT{p}$ the North Pole. A $(
Q_{\PT{a},s}, s )$-Fekete set minimizes the discrete weighted energy
associated with $Q_{\PT{a},s}$,
\begin{equation*}
E_s^{Q_{\PT{a},s}}( \{ \PT{x}_1, \PT{x}_2, \PT{x}_3, \PT{x}_4 \} ) \DEF
\mathop{\sum_{j=1}^4 \sum_{k=1}^4}_{j \neq k} \Big[ k_s( \PT{x}_j,
\PT{x}_k ) + Q_{\PT{a},s}(\PT{x}_j) + Q_{\PT{a},s}(\PT{x}_k) \Big],
\end{equation*}
among all four point configurations $\{ \PT{x}_1, \PT{x}_2, \PT{x}_3, \PT{x}_4 \}$ on $\mathbb{S}^2$. This optimization
problem is highly non-linear and currently eludes explicit solution. Heuristic considerations based on the symmetries of the problem and backed by numerical experiments suggest three basic types of optimal configurations (see Figure~\ref{fig2}): {\bf (A)} a triangular pyramid $X_{\{1,3\}}$ with one point at the South Pole and three points
forming an equilateral triangle parallel to the equator (one degree of freedom) with discrete weighted energy
\begin{equation*}
E_s^{Q_{\PT{a},s}}( X_{\{1,3\}} ) = f_{\{1,3\}}( t ) \DEF 6 \left( \frac{2^{-s/2}}{\left( 1 + t \right)^{s/2}} + \frac{3^{-s/2}}{\left( 1 - t^2 \right)^{s/2}} + q \left[ \frac{1}{\left( 1 + R \right)^s} + \frac{3}{\left( 1 - 2 R t + R^2 \right)^{s/2}} \right] \right),
\end{equation*}
where $t$ denotes the intercept of the triangle's plane with the polar axis;
{\bf (B)} the set $X_{\{2,2\}}$ consisting of two pairs
of opposite points in planes parallel to the equator rotated by $90^\circ$ (two degrees of freedom) with discrete weighted energy
\begin{equation*}
\begin{split}
E_s^{Q_{\PT{a},s}}( X_{\{2,2\}} ) = f_{\{2,2\}}( t, \tau )
&\DEF \frac{2^{1-s}}{\left( 1 - t^2 \right)^{s/2}} + \frac{2^{1-s}}{\left( 1 - \tau^2 \right)^{s/2}} + \frac{2^{3-s/2}}{\left( 1 - t \tau \right)^{-s/2}} \\
&\phantom{=}+ 12 q \left( \frac{1}{\left( 1 - 2 R t + R^2 \right)^{s/2}} + \frac{1}{\left( 1 - 2 R \tau + R^2 \right)^{s/2}} \right),
\end{split}
\end{equation*}
where $t$ and $\tau$ denote the intercepts of these two planes with the polar axis;
and {\bf (C)} the four points forming a square $X_{\{0,4\}}$ parallel to
the equator (one degree of freedom) with discrete weighted energy
\begin{equation*}
E_s^{Q_{\PT{a},s}}( X_{\{0,4\}} ) = f_{\{0,4\}}( t )
\DEF \frac{2^{2-s}\left( 1 + 2^{1+s/2} \right)}{\left( 1 - t^2 \right)^{s/2}} + \frac{24 q}{\left( 1 - 2 R t + R^2 \right)^{s/2}},
\end{equation*}
where $t$ is the intercept of the square's plane with the polar axis.

Using elementary calculus one can show that $f_{\{0,4\}}( t )$ is a strictly convex function on $(-1,1)$ with $f_{\{0,4\}}^\prime( t ) \to \pm \infty$ as $t \to \pm 1$ and thus has a unique minimum at a $t \in (-1,1)$ satisfying
\begin{equation*}
2^{2-s} \left( 1 + 2^{1+s/2} \right) s \, t \left( 1 - t^2 \right)^{-s/2-1} + 24 q s R \left( 1 - 2 R t + R^2 \right)^{-s/2-1} = 0.
\end{equation*}
Similarly, the function $f_{\{1,3\}}( t )$ is strictly convex on $(-1,1)$ with $f_{\{1,3\}}^\prime( t ) \to \pm \infty$ as $t \to \pm 1$ and thus has a unique minimum at a $t \in (-1,1)$ satisfying
\begin{equation*}
6 s \left( -2^{-s/2-1} \left( 1 + t \right)^{-s/2-1} + 3 q R \left( 1 - 2 R t + R^2 \right)^{-s/2-1} + 3^{-s/2} t \left( 1 - t^2 \right)^{-s/2-1}  \right) = 0.
\end{equation*}

The set $X_{\{0,4\}}$ is a special case of $X_{\{2,2\}}$ when the two planes merge into one. Figure~\ref{fig1} shows a comparison of the discrete weighted energy for these configurations for the Coulomb case $s = 1$ and for $q = 1/3$ and $q = 1$ as the distance of the external source varies. (As we have no explicit formulas, numerics provide the optimal values of the free parameters, that is, the positions of the planes supporting the points of $X_{\{1,3\}}$, $X_{\{2,2\}}$ and $X_{\{0,4\}}$ that yield the smallest respective weighted energy.)
\begin{figure}[ht]
\includegraphics[scale=.875]{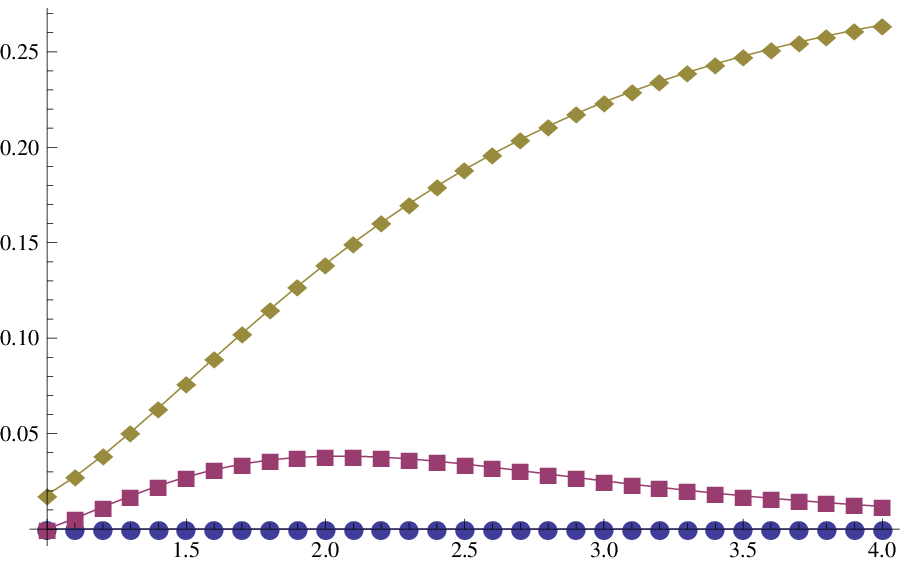}
\includegraphics[scale=.875]{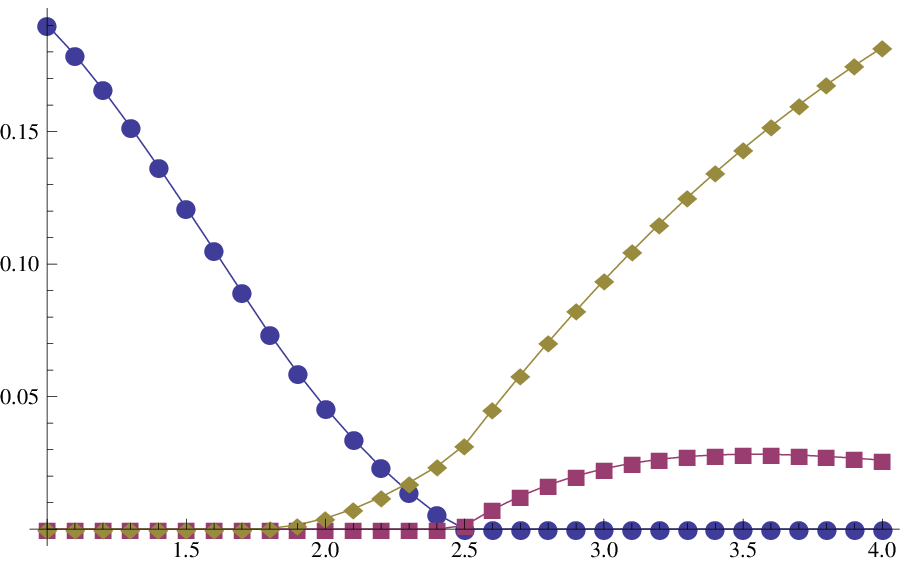}
\caption{\label{fig1} Deviation of the minimal discrete weighted energy of $X_{\{1,3\}}$ ($\bullet$), $X_{\{2,2\}}$ ($\blacksquare$) and $X_{\{0,4\}}$ ($\blacklozenge$) from the putative global minimum versus distance from center $R$ (sampled at $R_k = 1 + k / 10$) for $q = 1/3$ (left) and $q = 1$ (right) and Coulomb case $s = 1$.}
\end{figure}
In the case when $q = 1/3$ (weak field), the (numerically) smallest
discrete weighted energy appears to occur for the $X_{\{1,3\}}$-configuration for all $R \geq 1$, whereas for $q = 1$
(strong field) the optimal configurations change type from $X_{\{1,3\}}$ to $X_{\{2,2\}}$ to $X_{\{0,4\}}$ as the distance $R$ decreases and passes certain critical distances; cf. Figure~\ref{fig2} for typical examples.

\begin{rmk}
Similar to the case of the three point problem, the choice $Q_{\PT{p},s}( \PT{x} ) = \frac{1}{3} | \PT{x} - \PT{p} |^{-s}$ for the four point problem corresponds to the equilibrium configuration for Riesz-$s$ energy with no external field for \emph{five} points on $\mathbb{S}^2$. For $s = 1$ this is illustrated in the first frame of Figure~\ref{fig2}. Recently, Schwartz~\cite{Sch2013} gave a computer-assisted proof that the triangular bi-pyramid $X_{\{1,3\}} \cup \{ \PT{p} \}$ is indeed the minimizing configuration for $s = 1$ and $s = 2$. For other choices of $q$ in $Q_{\PT{p},s}( \PT{x} ) = q | \PT{x} - \PT{p} |^{-s}$ different configurations, such as the square based pyramid, may arise; see the second frame of Figure~\ref{fig2}.
\end{rmk}

\begin{figure}[ht]
\includegraphics[scale=0.4]{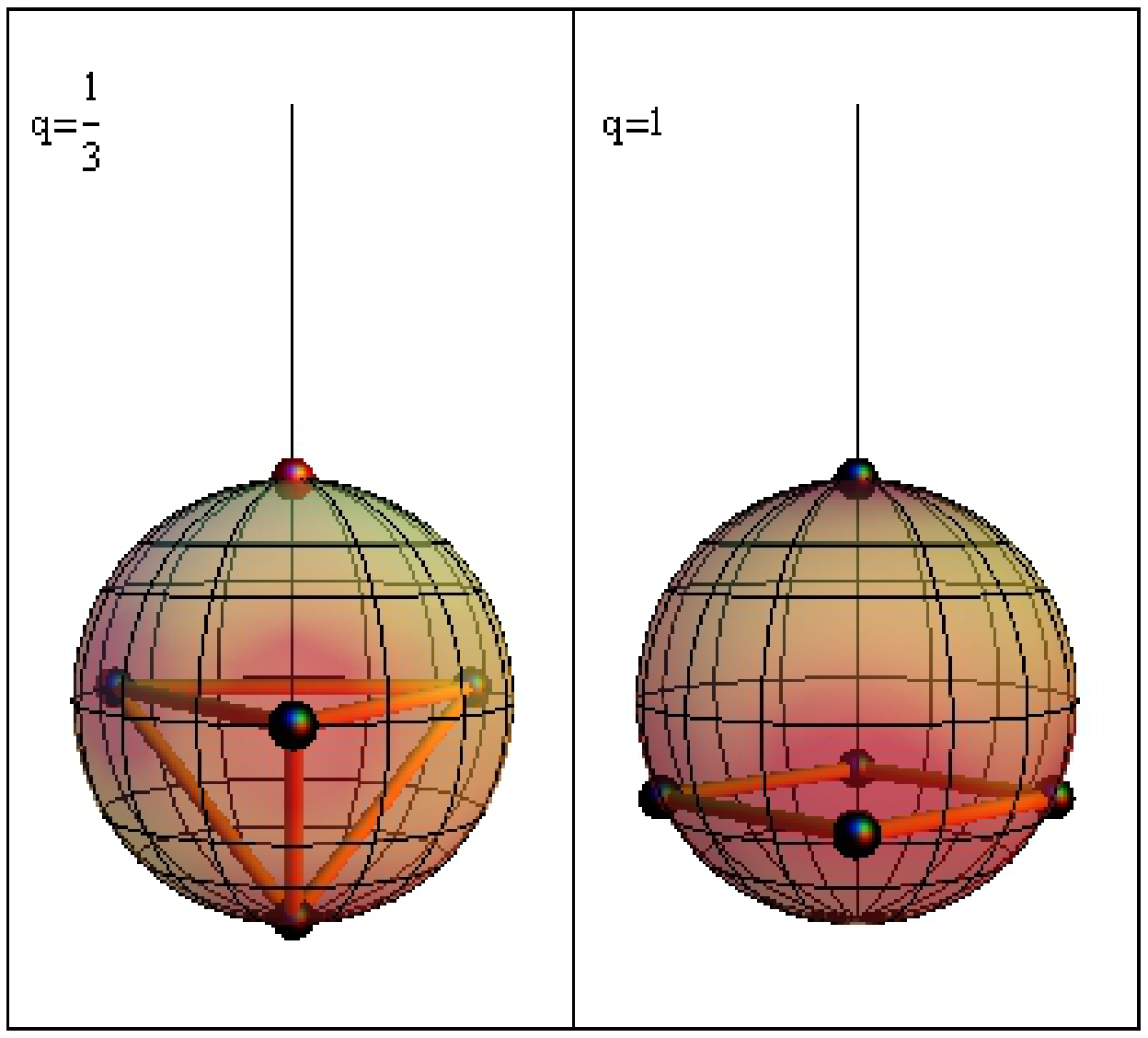}
\includegraphics[scale=0.4]{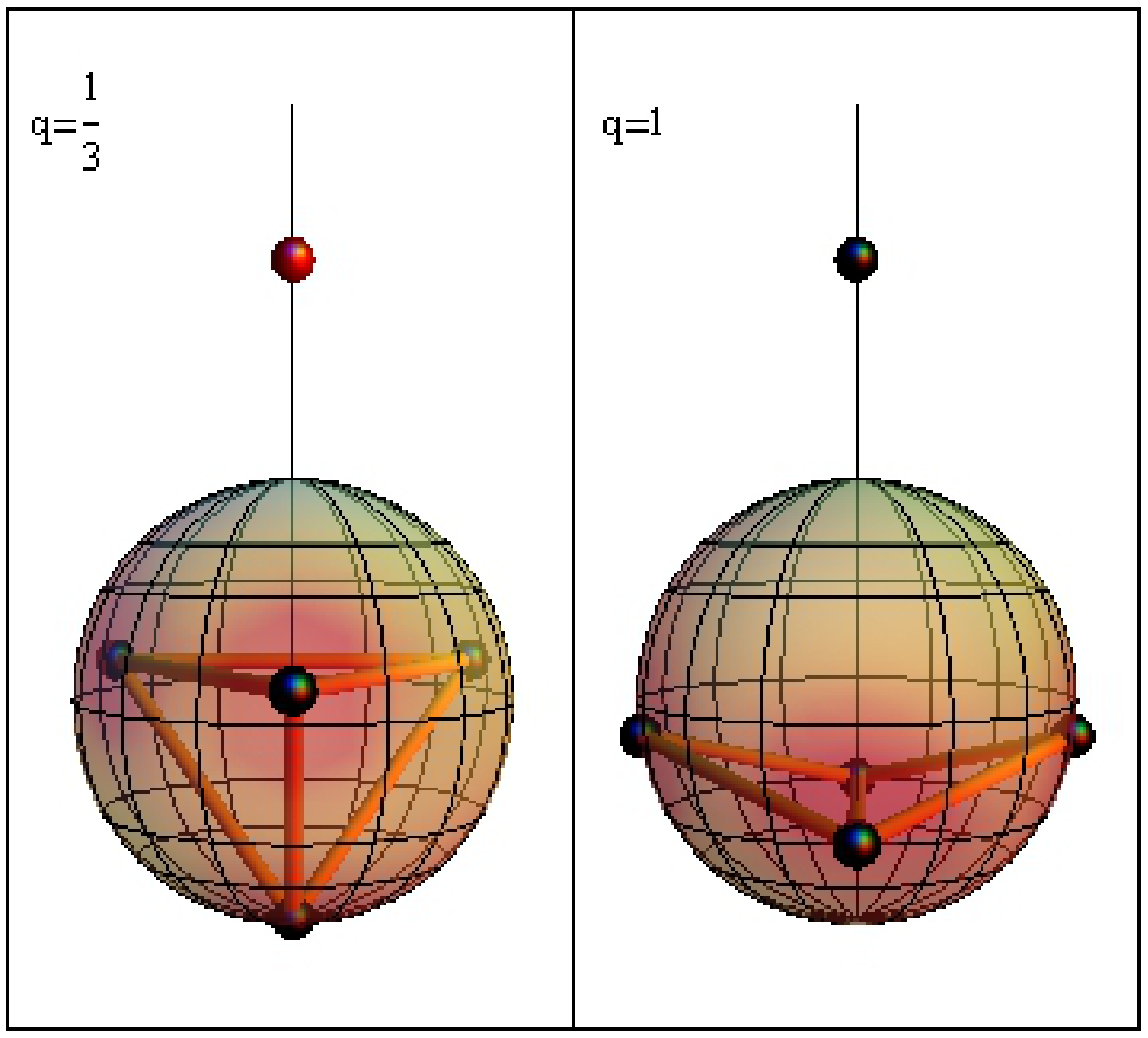}
\includegraphics[scale=0.4]{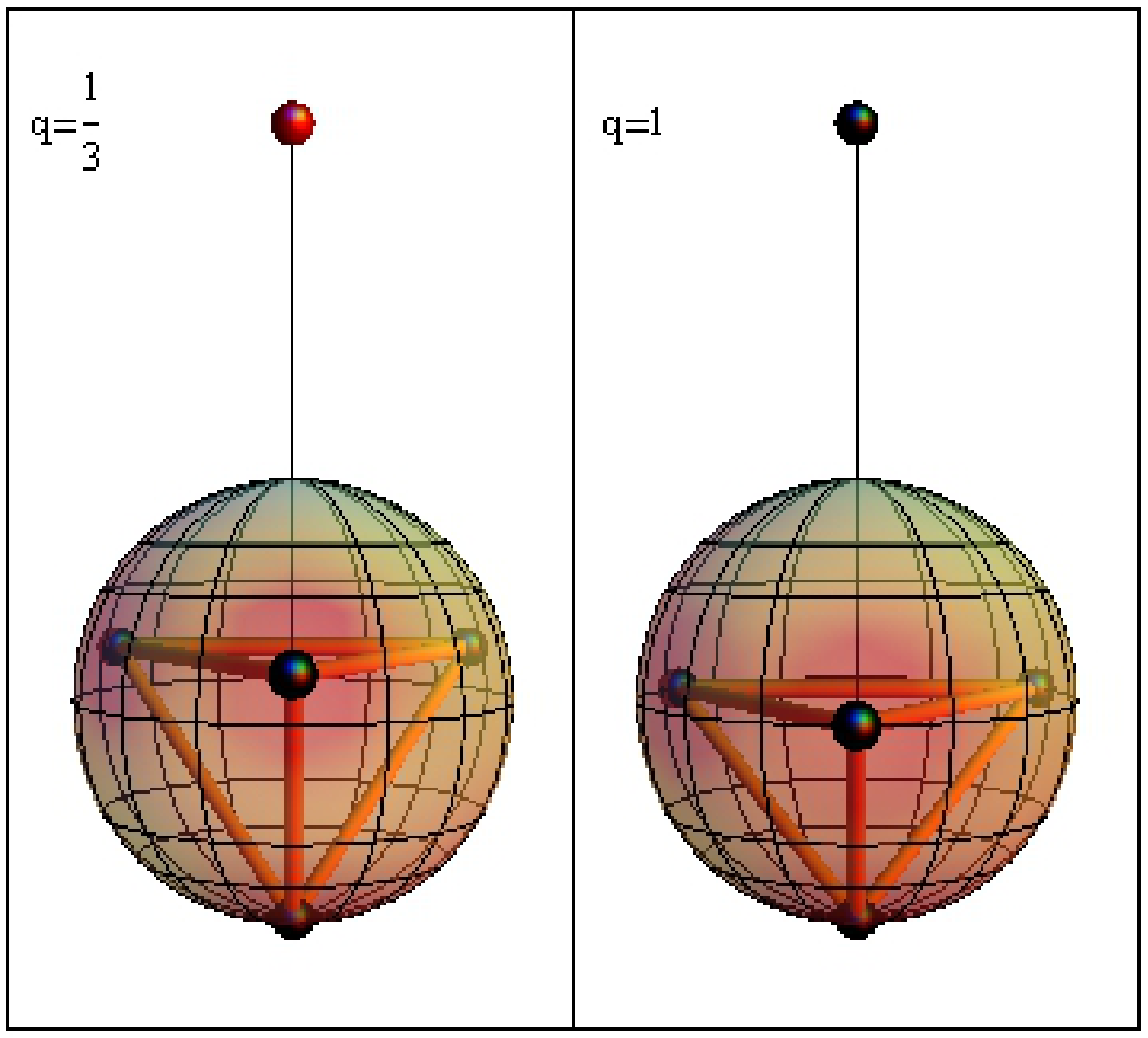}
\caption{\label{fig2} Typical {\bf putative} $4$-point $(Q_{\PT{a},s},s)$-Fekete sets for weaker (${q = 1/3}$, left) and stronger (${q = 1}$, right) Riesz external fields $Q_{\PT{a},s}( \PT{x} ) = q | \PT{x} - \PT{a} |^{-s}$, $\PT{a} = R \PT{p}$, in the Coulomb case $s = 1$ for three choices of center distances $R = 1$, $2$, and $1 + \phi$ ($\phi$ is the Golden ratio); compare with Fig.~\ref{fig3}.}
\end{figure}

\section{Proofs}
\label{sec:proofs}

\subsection{Proofs of Section~\ref{sec:characterization}}
\label{subsec:proofs.characterization}

The proof of Theorem~\ref{thm:main} relies on the restricted version of the maximum principle for measures supported on $\mathbb{S}^d$ (see Theorem~\ref{thm:restr.max.principle}). The latter can be shown using the principle of domination and the maximum principle for Riesz potentials.
For convenience we state them here.

\begin{prop}[Principle of Domination, {\cite[Thm.~1.29]{La1972}}] \label{prop:Thm1.29}
Let $p-2 \leq s < p$. Suppose $\mu$ is a positive measure in $\mathbb{R}^p$ whose potential $U_s^\mu$ is finite $\mu$-almost everywhere, and that $f(\PT{x})$ is a $(p-s)$-superharmonic function. Then if the inequality $U_s^\mu(\PT{x})\leq f(\PT{x}) $ holds $\mu$-almost everywhere, it holds everywhere in $\mathbb{R}^p$.
\end{prop}

\begin{prop}[Maximum Principle, {\cite[Thm.~1.10]{La1972}}]
Let $p-2 \leq s < p$. Suppose $\mu$ is a positive measure in $\mathbb{R}^p$ and that $U_s^\mu(\PT{x}) \leq M$ holds $\mu$-almost everywhere for some real $M$. Then this inequality holds throughout all of $\mathbb{R}^p$.
\end{prop}

Before we provide the proof of Theorem~\ref{thm:restr.max.principle} and then show Theorem~\ref{thm:main}, we recall some basic facts about the Kelvin transformation (spherical inversion) of points and measures. Inversion in a sphere is a basic technique in electrostatics (method of electrical images, cf. \cite{Ja1998}) and, more generally, in potential theory (cf. \cite{Ke1967} and \cite{La1972}). The Kelvin transformation (of a function) is linear, preserves harmonicity, and preserves positivity. Let us denote by $\kelvin_{\PT{a}}$ the Kelvin transformation (stereographic projection) with center $\PT{a} \in \mathbb{S}^d$ and radius $\sqrt{2}$; that is, for any point $\PT{x} \in \mathbb{R}^{d+1}$ the image $\PT{x}^* \DEF
\kelvin_{\PT{a}}(\PT{x})$ lies on a ray emanating from $\PT{a}$, and passing through $\PT{x}$ such that
\begin{equation} \label{KelTr}
\left| \PT{x} - \PT{a} \right| \cdot \left| \PT{x}^* - \PT{a} \right| = 2.
\end{equation}
The transformation of the distance is given by the formula
\begin{equation} \label{KelTrDist}
\left| \PT{x}^{*} - \PT{y}^{*} \right| = 2 \frac{\left| \PT{x} -
\PT{y} \right|}{\left| \PT{x} - \PT{a} \right| \left| \PT{y} -
\PT{a} \right|}, \qquad \PT{x},\PT{y}\in\mathbb{S}^{d}.
\end{equation}
The image of $\mathbb{S}^d$ under the Kelvin transformation is a hyperspace orthogonal to the radius-vector $\PT{a}$, which we can identify with $\mathbb{R}^d$ in a natural way. For $A \subset \mathbb{S}^d$, let $A^*$ denote the image of $A$ under the Kelvin transformation $\kelvin_{\PT{a}}$. Clearly, $A^*$ is then contained in $\mathbb{R}^d$.

Next, we recall the definition of the Kelvin transform of measures. Given a measure $\nu$ with no point mass at $\PT{a}$, its {\em $s$-Kelvin transformation} $\nu^* = \kelvinMEAS_{\PT{a},s}(\nu)$ is a measure defined by
\begin{equation}
\dd \nu^* (\PT{x}^* ) \DEF \frac{2^{s/2}}{\left| \PT{x} - \PT{a}
\right|^s} \dd \nu(\PT{x}). \label{KelTr1}
\end{equation}
Clearly, \eqref{KelTr} and \eqref{KelTr1} imply the duality $(\nu^*)^* =\nu$. We also note that
\begin{equation} \label{KelTrPot}
U^{\nu^*}_s(\PT{x}^*)=\frac{2^{s/2}}{|\PT{x}^*-\PT{a}|^s} \, U^\nu_s(\PT{x}), \qquad \PT{x} \in \mathbb{S}^d.
\end{equation}
In the particular case when $d=2$ and $s=\log$ we have
\begin{equation} \label{KelTr1_log}
\dd \nu^* (\PT{x}^* ) \DEF \dd \nu(\PT{x}), \qquad U^{\nu^*}_{\log}(\PT{x}^*)=U^\nu_{\log}(\PT{x})-U^\nu_{\log}(\PT{a}) +\log|\PT{x}-\PT{a}| \qquad \PT{x} \in \mathbb{S}^2.
\end{equation}

\begin{proof}[Proof of Theorem~\ref{thm:restr.max.principle}]
Suppose that $d-2 < s < d$ or if $s=d-2$, then $d\geq 3$. Any point~$\PT{a}$ in the level set $\{ \PT{x} \in \mathbb{S}^d : U_s^\mu(\PT{x}) \leq M \}$, which has positive $\mu$-measure, can serve as a center of inversion for the Kelvin transformation with radius $\sqrt{2}$ given in \eqref{KelTr}. Select one such point as center. Define the measure $\nu \DEF [ M / W_s(\mathbb{S}^d) ] \sigma_d$. Then $U_s^\nu(\PT{x}) = M$ for all $\PT{x}\in\mathbb{S}^d$. Neither $\mu$ nor $\nu$ has a point mass at $\PT{a}$, hence \eqref{KelTrPot} gives
\begin{equation*}
U_s^{\mu^*}(\PT{x}^*) =\frac{2^{s/2}}{|\PT{x}^*-\PT{a}|^s} \, U^\mu_s (\PT{x}) \leq \frac{M2^{s/2}}{|\PT{x}^*-\PT{a}|^s} = U_s^{\nu^*} (\PT{x}^*)  \qquad \text{$\mu$-a.e. (and hence $\mu^*$-a.e.)}.
\end{equation*}
Since the right-hand side is $(d-s)$-superharmonic in $\mathbb{R}^d$, by Proposition~\ref{prop:Thm1.29} the inequality
extends to all $\PT{x}^*\in\mathbb{R}^d = (\mathbb{S}^d)^*$. The inverse Kelvin transformation yields $U_s^\mu(\PT{x}) \leq U_s^\nu(\PT{x})$ for all $\PT{x} \in \mathbb{S}^d \setminus \{ \PT{a} \}$. So, $U_s^\mu(\PT{x}) \leq M$ for all $\PT{x} \in \mathbb{S}^d \setminus \{ \PT{a} \}$. Using a different center of inversion, one obtains $U_s^\mu(\PT{a}) \leq M$.

Suppose now that $d = 2$ and $s = \log$. Using \eqref{KelTr1_log} we write
\begin{equation*}
\begin{split}
U_{\log}^{\mu^*}(\PT{x}^*) + \log|\PT{x}^*-\PT{a}| + U_{\log}^{\mu}(\PT{a})
&= U^{\mu_{\log}} (\PT{x}) \\
&\leq M+ U_{\log}^{\sigma_2} (\PT{x})-W_{\log}(\mathbb{S}^2) \\
&= U_{\log}^{\sigma_2^*}(\PT{x}^*)+ \log|\PT{x}^*-\PT{a}| + c_2  \qquad \text{$\mu^*$-a.e.}
\end{split}
\end{equation*}
Since $\mu^*$ has finite logarithmic energy, the principle of domination for logarithmic potentials \cite[Theorem~II.3.2]{SaTo1997} implies that this inequality holds everywhere in the complex plane. The inverse Kelvin transformation then gives the desired inequality on the sphere except at the center of inversion. This restriction can be removed by moving the inversion center.
\end{proof}

\begin{proof}[Proof of Theorem~\ref{thm:main}]
Suppose to the contrary that there is a measure $\lambda \in
\mathcal{M}(\mathbb{S}^d)$ such that \eqref{eq:essinf} fails;
that is, there is a constant $L_1>F_{Q,s}$ such that
\begin{equation*}
U_s^\lambda(\PT{x}) + Q(\PT{x}) \geq L_1 \qquad \text{q.e. on $S_{Q,s}$.}
\end{equation*}
Applying \eqref{VarEq2} we obtain that
\begin{equation}
U_s^\lambda(\PT{x}) \geq U_s^{\mu_{Q,s}}(\PT{x}) + L_1 - F_{Q,s} \qquad
\text{q.e. on $S_{Q,s}$.} \label{eq:aux1}
\end{equation}
From Proposition~\ref{prop:1} we have that $\mu_{Q,s}$ has finite
$s$-energy, therefore its support $S_{Q,s}$ will have positive
$s$-capacity. Hence, the $s$-extremal measure associated with $S_{Q,s}$ is
well-defined. Thus, we may integrate both sides of inequality
\eqref{eq:aux1} with respect to $\mu_{S_{Q,s}}$ and, using Fubini's
theorem, we derive
\begin{equation}
\begin{split}
\int U_s^{\mu_{S_{Q,s}}} \dd\lambda = \int U_s^\lambda \dd\mu_{S_{Q,s}}
&\geq \int U_s^{\mu_{Q,s}} \dd \mu_{S_{Q,s}} + L_1 - F_{Q,s} \\
&= \int U_s^{\mu_{S_{Q,s}}} \dd \mu_{Q,s} + L_1 - F_{Q,s}.
\end{split}
\end{equation}
Recall that the Gauss variational inequalities for the
$s$-extremal measure $\mu_{S_{Q,s}}$ state that (see
\cite{La1972})
\begin{equation*}
U_s^{\mu_{S_{Q,s}}}(\PT{x}) \leq W_s(S_{Q,s}) \quad \text{on $\supp
(\mu_{S_{Q,s}})$}, \qquad U_s^{\mu_{S_{Q,s}}}(\PT{x}) \geq
W_s(S_{Q,s}) \quad \text{q.e. on $S_{Q,s}$.}
\end{equation*}
By the maximum principle the first inequality can be extended to all
of $\mathbb{R}^{d+1}$ if $d-1 \leq s < d$ and still remains true on
$\mathbb{S}^d$ when $d-2 \leq s < d-1$ by the sphere maximum
principle (see Theorem~\ref{thm:restr.max.principle}). Thus, we get
the contradiction
\begin{equation*}
W_s(S_{Q,s}) \geq W_s(S_{Q,s}) + L_1 - F_{Q,s} > W_s(S_{Q,s}).
\end{equation*}
This establishes relation \eqref{eq:essinf} for $d-2\leq s<d$.

We derive \eqref{eq:sup} similarly, utilizing \eqref{VarEq1}
instead. Assume there is a measure $\lambda \in
\mathcal{M}(\mathbb{S}^d)$ and a constant $L_2 < F_{Q,s}$, such that
\begin{equation*}
U_s^\lambda(\PT{x}) + Q(\PT{x}) \leq L_2, \qquad \PT{x} \in \supp
(\lambda).
\end{equation*}
Integration with respect to $\lambda$ yields that $\lambda$ has
finite $s$-energy (recall that $F_{Q,s}$ is finite and $Q(\PT{x})
\geq c_Q$). This implies that ${\rm cap}_s(\supp (\lambda)) > 0$,
and that the measure $\mu_{\supp(\lambda)}$ is well-defined. From
\eqref{VarEq1} we get
\begin{equation*}
U_s^{\mu_{Q,s}}(\PT{x}) \geq U_s^\lambda(\PT{x}) + F_{Q,s} - L_2 \qquad
\text{q.e. on $\supp(\lambda)$.}
\end{equation*}
Now, we integrate both sides of this inequality with respect to the
$s$-extremal measure $\mu_{\supp(\lambda)}$ and using Fubini's
theorem we arrive at
\begin{equation*}
\begin{split}
\int U_s^{\mu_{\supp(\lambda)}} \dd \mu_{Q,s} = \int U_s^{\mu_{Q,s}} \dd
\mu_{\supp(\lambda)}
&\geq \int U_s^\lambda \dd \mu_{\supp(\lambda)} + F_{Q,s} - L_2 \\
&= \int U_s^{\mu_{\supp(\lambda)}} \dd \lambda + F_{Q,s} - L_2.
\end{split}
\end{equation*}
Applying again the Gauss variational inequalities now to
$U_s^{\mu_{\supp(\lambda)}}$ and using the sphere maximum principle
we obtain a similar contradiction.
\end{proof}

\begin{proof}[Proof of Theorem~\ref{thm:signsupp.2}]
Let $d - 2 \leq s < d$ with $s > 0$. For brevity set $\eta = \eta_{\Sigma_\rho, Q, s}$ and $\sgnEqconst_\eta = \sgnEqconst_{\Sigma_\rho, Q, s}$, where $\Sigma_\rho$ is the spherical cap $\{ \PT{x} \in \mathbb{S}^d : | \PT{x} - \PT{p} | \geq \rho \}$. Then
\begin{equation*}
U_s^{\eta^+}(\PT{x}) - U_s^{\eta^-}(\PT{x}) + Q(\PT{x}) = \sgnEqconst_\eta \qquad \text{everywhere on $\Sigma_\rho$.}
\end{equation*}
Using \eqref{VarEq1} and \eqref{estlspt} we get
\begin{align}
U_s^{\mu_Q+\eta^-}(\PT{x})+\sgnEqconst_\eta &\geq U_s^{\eta^+}(\PT{x})+F_{Q,s} \quad \text{q.e. on $\mathbb{S}^d \cap \Sigma_\rho$,} \label{VarEq5.b} \\
U_s^{\mu_Q+\eta^-}(\PT{x})+\sgnEqconst_\eta &\leq U_s^{\eta^+}(\PT{x})+F_{Q,s} \quad \text{everywhere on $\widetilde{S}_{Q,s} \cap \Sigma_\rho$.} \label{VarEq6.b}
\end{align}
Now we can rewrite \eqref{VarEq5.b} as
\begin{equation*}
U_s^{\mu_{Q,s}+\eta^-}(\PT{x}) + (\sgnEqconst_\eta - F_{Q,s}) \CAP_s(\mathbb{S}^d) U_s^{\bal_s( \sigma_d, \Sigma_\rho )}(\PT{x}) \geq U_s^{\eta^+}(\PT{x}) \qquad \text{q.e. on $\mathbb{S}^d \cap \Sigma_\rho$}
\end{equation*}
which also holds $\eta^+$-a.e. on $\mathbb{S}^d \cap \Sigma_\rho$.
Setting $\nu \DEF \mu_{Q,s} + \eta^- + (\sgnEqconst_\eta - F_{Q,s}) \CAP_s(\mathbb{S}^d) \bal_s( \sigma_d, \Sigma_\rho )$, we obtain that (see also \eqref{VarEq6.b})
\begin{equation*}
U^\nu_s (\PT{x}) \geq U_s^{\eta^+}(\PT{x}) \quad \text{$\eta^+$-a.e. on $\mathbb{S}^d \cap \Sigma_\rho$,} \qquad U^\nu_s (\PT{x}) \leq U_s^{\eta^+}(\PT{x}) \quad \text{everywhere on $\widetilde{S}_{Q,s} \cap \Sigma_\rho$.}
\end{equation*}
Selecting a center of inversion $\PT{a} \in \mathbb{S}^{d} \setminus \supp( \eta^+ )$, the potentials of the Kelvin transformations $\nu^*$ and $(\eta^+)^*$ satisfy
\begin{align*}
U^{\nu^*}_s (\PT{x}^*) &\geq U_s^{(\eta^+)^*}(\PT{x}^*) \qquad \text{$(\eta^+)^*$-a.e. on $(\mathbb{S}^d \cap \Sigma_\rho)^*$,} \\
U^{\nu^*}_s(\PT{x}^*)  &\leq U_s^{(\eta^+)^*}(\PT{x}^*) \qquad \text{everywhere on $(\widetilde{S}_{Q,s} \cap \Sigma_\rho)^*$.}
\end{align*}
Note that by the principle of domination (see Proposition \ref{prop:Thm1.29}) the first inequality holds on all of $\mathbb{R}^d$.
Now, we can apply a de La Vall\'{e}e Poussin-type theorem (see
\cite[Section~3]{Fu1992}, \cite[Theorem~2.5]{Ja2000}) to conclude that
\begin{equation*}
\nu^* \big|_{(\widetilde{S}_{Q,s} \cap \Sigma_\rho)^*} \leq (\eta^+)^* \big|_{(\widetilde{S}_{Q,s} \cap \Sigma_\rho)^*}.
\end{equation*}
The inverse Kelvin transformation yields
\begin{equation*}
\left( \mu_{Q,s} + \eta^- + (\sgnEqconst_\eta - F_{Q,s}) \CAP_s(\mathbb{S}^d) \bal_s( \sigma_d, \Sigma_\rho ) \right)\big|_{\widetilde{S}_{Q,s} \cap \Sigma_\rho} \leq \eta^+\big|_{\widetilde{S}_{Q,s} \cap \Sigma_\rho}.
\end{equation*}
This implies that
\begin{equation*}
\mu_{Q,s} \big|_{\Sigma_\rho} \leq \eta^+ \big|_{S_{Q,s}}, \qquad S_{Q,s} \cap \Sigma_\rho \subset \supp(\eta^+),
\end{equation*}
and that if $F_{Q,s} < \sgnEqconst_\eta$, then $\widetilde{S}_{Q,s} \cap \Sigma_\rho \subset \supp(\eta^+)$. The theorem for the Riesz case $d - 2 \leq s < d$ with $s > 0$ is proved.

When $d=2$ and $s=\log$ we utilize \eqref{KelTr1_log} and modify the argument above using the regular principle of domination in the complex plane and the original de La Vall\'{e}e Poussin theorem \cite[Theorem~IV.4.5]{SaTo1997}.
\end{proof}

Next we establish the proof of Theorem \ref{thm:main2}, which allows
for the important conclusion that the $(Q,s)$-Fekete sets
$\PSET_{n,Q,s}$ (that is, the supports of the minimizers of the
discrete weighted $s$-energy associated with $Q$) are contained in
the extended support $\widetilde{S}_{Q,s}$ of the continuous
minimizer $\mu_{Q,s}$.

\begin{proof}[Proof of Theorem~\ref{thm:main2}]
From Theorem \ref{thm:main} we conclude that $M \leq F_{Q,s}$. Using \eqref{VarEq2}, inequality \eqref{th7.1} yields that
\begin{equation*}
U^{\mu_{\PSET_n}}_s( \PT{x} ) + F_{Q,s} - M \geq U_s ^{\mu_{Q,s}}(
\PT{x} ) \qquad \text{q.e. on $S_{Q,s}$.}
\end{equation*}
Let $\lambda$ be a multiple of the $s$-equilibrium measure on $\mathbb{S}^d$, so that $U_s^\lambda(\PT{x}) = F_{Q,s} - M$ for all $\PT{x} \in \mathbb{S}^d$. Then the last inequality becomes
\begin{equation} \label{eq:inequality.A}
U_s^{\mu_{\PSET_n}+\lambda}(\PT{x}) \geq U_s^{\mu_{Q,s}}(\PT{x}) \qquad \text{q.e. on $S_{Q,s}$.}
\end{equation}

If $S_{Q,s} = \mathbb{S}^d$, then we can extend this inequality to any $\PT{z} \not\in \PSET_n$ by the lower semi-continuity of $U_s^{\mu_{Q,s}}$ and the continuity of $U_s^{\mu_{\PSET_n}+\lambda}$ at such $\PT{z}$. As $U_s^{\mu_{\PSET_n}+\lambda}(\PT{x}) = + \infty$ and $U_s^{\mu_{Q,s}}(\PT{x}) < \infty$ for $\PT{x} \in \PSET_n$, \eqref{th7.2} holds for all $\PT{x}\in \mathbb{S}^d$.

If $S_{Q,s}$ is a proper subset of $\mathbb{S}^d$, then we can find a point $\PT{a} \not\in \PSET_n \cap S_{Q,s}$. Since $\mu_{Q,s}$ has finite
$s$-energy, inequality~\eqref{eq:inequality.A} holds $\mu_{Q,s}$-almost everywhere. Using
Kelvin transform centered at $\PT{a}$ with radius $\sqrt{2}$, we derive that
\begin{equation*}
U_s^{(\mu_{\PSET_n}+\lambda)^*}(\PT{x}^*) \geq
U_s^{\mu^*_{Q,s}}(\PT{x}^*) \qquad \text{$\mu_{Q,s}^*$-a.e.}
\end{equation*}
The principle of domination (see Proposition \ref{prop:Thm1.29}) enables us to extend this inequality to all $\PT{x}^* \in \mathbb{R}^d$. Inverse Kelvin
transformation then implies \eqref{th7.2} for all $\PT{x}\in \mathbb{S}^d$ except at the center $\PT{a}$ and this restriction can be lifted by moving the center of inversion.

Finally, we note that (\ref{th7.3}) is an immediate consequence of
\eqref{th7.2} and \eqref{VarEq1}.
\end{proof}

\begin{proof}[Proof of Corollary~\ref{cor}]
Definition~\ref{def:discrete.external.field.problem} implies that any point $\PT{x}_k \in \PSET_{n,Q,s}$ is a global minimum for the weighted potential (cf. \eqref{DiscrPot}) $h_{\PSET_{n-1}}(\PT{x})$, $\PT{x} \in \mathbb{S}^d$, where $\PSET_{n-1} = \PSET_{n,Q,s} \setminus \{\PT{x}_k\}$. Hence, inequality~\eqref{th7.1} holds with $M = h_{\PSET_{n-1}}(\PT{x}_k) = U_s^{\mu_{\PSET_{n-1}}}(\PT{x}_k) + Q(\PT{x}_k)$. Thus \eqref{th7.2} holds in particular for $\PT{x} = \PT{x}_k$ which reduces to $U_s^{\mu_{Q,s}}(\PT{x}_k) + Q(\PT{x}_k) \leq F_{Q,s}$. Therefore, $\PT{x}_k \in \widetilde{S}_{Q,s}$.
\end{proof}

\subsection{Proofs of Section~\ref{sec:separation}}
Here we prove our separation results. The Proof of Theorem~\ref{thm:main3} has three parts. The first considers the case $d - 2 < s < d$. The second establishes explicit bounds in the limiting case $s = d - 2$ in the field-free setting and thus shows Proposition~\ref{prop:well-separation.min.Riesz.(d-2).energy}. The third considers the limiting case $s = d - 2 > 0$ in the presence of an external field.

\begin{proof}[Proof of Theorem~\ref{thm:main3}, Part~1]
Let $d - 2 < s < d$. We shall modify the approach from \cite{DrSa2007}. Let $\PSET_{n,Q,s} = \{ \PT{x}_1, \dots, \PT{x}_n\}$ be an $n$-point $(Q,s)$-Fekete set on $\mathbb{S}^d$. Each point in $\PSET_{n,Q,s}$ defines an external field by means of 
\begin{equation*}
\widetilde{Q}_k(\PT{x}) \DEF Q(\PT{x}) + \frac{1}{(n-2)|\PT{x}-\PT{x}_k|^s}, \quad \PT{x} \in \mathbb{S}^d, \qquad k = 1, \dots, n.
\end{equation*}
Observe that $\PSET_{n,Q,s} \setminus \{\PT{x}_k\}$ is an $(n-1)$-point $(\widetilde{Q}_k,s)$-Fekete set on~$\mathbb{S}^d$.
Since all the $m$-point $(\widetilde{Q}_k,s)$-Fekete sets $\PSET_{m,\widetilde{Q}_k,s}$ are contained in the extended support $\widetilde{S}_{\widetilde{Q}_k,s}$ of the $s$-extremal measure $\mu_{\widetilde{Q}_k,s}$ on $\mathbb{S}^d$ associated with $\widetilde{Q}_k$ by Corollary~\ref{cor}, we will be done if we show that there is a spherical cap with radius $K_{Q,s} / n^{1/d}$ such that when centered at $\PT{x}_k$ its intersection with $\widetilde{S}_{\widetilde{Q}_k,s}$ is empty. Note that the constant $K_{Q,s}$ should only depend on $Q$ and $s$.

We shall use the fact that whenever the compact subset $K \subset \mathbb{S}^d$ contains the support $S_{\widetilde{Q}_k,s}$ of the $s$-extremal measure $\mu_{\widetilde{Q}_k,s}$ on $\mathbb{S}^d$, then $\widetilde{S}_{\widetilde{Q}_k,s} \subset \widetilde{S}_{K,\widetilde{Q}_k,s} \subset \supp( \eta_{K,\widetilde{Q}_k,s}^+ )$ (cf. Theorem~\ref{thm:signsupp.2} and remark following it).
From Proposition~\ref{prop:1}(b) it follows that $S_{\widetilde{Q}_k,s}$ is contained in a set where $\widetilde{Q}_k(\PT{x}) \leq M_k$ for some constant $M_k > 0$. This implies that there is a spherical cap $\Sigma_{r_k} \DEF \Sigma_{r_k}( \PT{x}_k ) \DEF \{ \PT{x} \in \mathbb{S}^d : |\PT{x} - \PT{x}_k| \geq r_k \}$, $r_k > 0$,
such that $S_{\widetilde{Q}_k,s} \subset \Sigma_{r_k}$.
Hence it suffices to consider the continuous Riesz external field problem on $\Sigma_{r_k}$ (instead of~$\mathbb{S}^d$) and study signed $s$-equilibria $\eta_{\Sigma_\rho,\widetilde{Q}_k,s}$ on spherical caps $\Sigma_{\rho}$ satisfying $S_{\widetilde{Q}_k,s} \subset \Sigma_\rho$.

The signed $s$-equilibrium $\eta_{\Sigma_\rho,\widetilde{Q}_k,s}$ on a spherical cap $\Sigma_\rho$ exists for any $2>\rho>0$. Indeed,
\begin{equation} \label{signedeqQ1}
\eta_{\Sigma_\rho, \widetilde{Q}_k, s} = c_k \bal_s( \sigma_d, \Sigma_\rho ) - \bal_s( \sigma, \Sigma_\rho ) - 1/(n-2)  \, \bal_s( \delta_{\PT{x}_k}, \Sigma_\rho),
\end{equation}
where $\delta_{\PT{x}_k}$ is the Dirac-delta measure with unit charge placed at $\PT{x}_k$ and $c_k$ is a normalizing constant such that $\| \eta_{\Sigma_\rho,\widetilde{Q}_k, s} \| = 1$.
Using similar analysis as in \cite[Proof of Theorem~1.5]{DrSa2007} we shall derive that this signed measure will be negative on a band containing the rim of the spherical cap $\Sigma_\rho$ for any $0< \rho < K_{Q,s} /n^{1/d}$, and in particular for $\rho=r_k$. Hence, for any such $\rho$ the support of the positive part $\eta_{\Sigma_\rho,\widetilde{Q}_k,s}^+$ is contained in a smaller spherical cap $\Sigma_{\tilde{\rho}}$ with $\rho < \tilde{\rho}$.
Since $\widetilde{S}_{\widetilde{Q}_k,s}$ is contained in any such $\supp( \eta_{\Sigma_\rho, \widetilde{Q}_k,s}^+ )$, we conclude that $\widetilde{S}_{\widetilde{Q}_k,s}\subset \Sigma_{K_{Q,s} / n^{1/d}}$, i.e.
\begin{equation*}
\dist( \PT{x}_k, \widetilde{S}_{\widetilde{Q}_k,s}) \geq K_{Q,s} / n^{1/d}.
\end{equation*}

To finish the proof we have to show the negativity of the signed equilibrium measures $\eta_{\Sigma_\rho, \widetilde{Q}_k, s}$ near the rim for all such $0< \rho < K_{Q,s} / n^{1/d}$. The middle term in \eqref{signedeqQ1}, $\bal_s( \sigma, \Sigma_\rho )$, can be written as
\begin{equation*}
\bal_s( \sigma, \Sigma_\rho ) = \bal_s( \sigma^+, \Sigma_\rho) - \bal_s( \sigma^-, \Sigma_\rho ) \FED \sigma^+_\rho - \sigma^-_\rho.
\end{equation*}
With the notations $\nu_\rho \DEF \bal_s( \sigma_d, \Sigma_\rho )$ and $\epsilon_\rho \DEF  \bal_s( \delta_{\PT{x}_k}, \Sigma_\rho)$ we rewrite \eqref{signedeqQ1} as follows:
\begin{align} \label{signedeqQ2}
\eta_{\Sigma_\rho, \widetilde{Q}_k, s} &= \frac{1+\|\sigma^+_\rho \|-\|\sigma^-_\rho \|+\|\epsilon_\rho\|/(n-2)}{\|\nu_\rho\|} \nu_\rho - \sigma^+_\rho  + \sigma^-_\rho - 1/(n-2)  \, \epsilon_\rho\nonumber \\
&\leq \frac{1+\|\sigma^+_\rho \|-\|\sigma^-_\rho \|+\|\epsilon_\rho\|/(n-2)}{\|\nu_\rho\|} \nu_\rho  + \sigma^-_\rho - 1/(n-2)\, \epsilon_\rho.
\end{align}
Note that inequality between two signed measures may be understood in terms of their densities or that the difference of the two signed measures is a positive measure.

Since balayage may be done in steps, we have
\begin{equation*}
\sigma_\rho^- = \bal_s( \sigma_0^-, \Sigma_\rho ), \quad  \text{where} \quad \sigma_0^- \DEF \bal_s( \sigma^- , \mathbb{S}^d ).
\end{equation*}
The superposition representation of the balayage measure yields
\begin{equation*}
\dd \sigma_0^- (\PT{x}) = \left( \int \epsilon_\PT{y}^\prime (\PT{x})\, d\sigma^- (\PT{y}) \right) \dd \sigma_d (\PT{x}), \qquad \text{where $\epsilon_\PT{y}=\bal_s( \delta_{\PT{y}}, \mathbb{S}^d)$, $\PT{y} \in \mathbb{R}^{d+1}$ with $| \PT{y} | > 1$.}
\end{equation*}
The formula for the point mass balayage on the unit sphere (see \cite[Theorem 2]{BrDrSa2009}) yields
\begin{align*}
\int \epsilon_\PT{y}^\prime (\PT{x})\, d\sigma^- (\PT{y})
&= \int \frac{\left(| \PT{y} |^2-1 \right)^{d-s}}{W_s(\mathbb{S}^d)|\PT{y}-\PT{x}|^{2d-s}}\, d\sigma^-(\PT{y}) \\
&\leq \int \frac{\left(| \PT{y} |^2-1 \right)^{d-s}}{W_s(\mathbb{S}^d)||\PT{y}|-1|^{2d-s}}\, d\sigma^-(\PT{y}) \leq \frac{(r+1)^{d-s}}{W_s(\mathbb{S}^d)(r-1)^d} \left\|\sigma^- \right\|,
\end{align*}
where $r > 1$ is such that $\supp( \sigma^- ) \subset \{ \PT{x} \in \mathbb{R}^{d+1} : |\PT{x}| \geq r \}$.
As balayage preserves positivity and is linear, application of balayage to the inequality
\begin{equation*}
\sigma_0^- (\PT{x}) \leq \frac{(r+1)^{d-s}}{W_s(\mathbb{S}^d)(r-1)^d} \left\|\sigma^- \right\| \sigma_d (\PT{x})
\end{equation*}
yields
\begin{equation} \label{signedeqQ4}
\sigma_\rho^- \leq \bal_s\Big( \frac{(r+1)^{d-s}}{W_s(\mathbb{S}^d)(r-1)^d} \left\|\sigma^- \right\| \sigma_d (\PT{x}), \Sigma_\rho \Big)  = \frac{(r+1)^{d-s} \|\sigma^-\|}{W_s(\mathbb{S}^d)(r-1)^d} \, \nu_\rho.
\end{equation}

Combining \eqref{signedeqQ2} and \eqref{signedeqQ4}, and using that $\|\nu_\rho\|\leq 1$, we arrive at
\begin{equation*} 
\eta_{\Sigma_\rho, \widetilde{Q}_k, s} \leq \frac{1+\|\sigma^+_\rho \|+\left(\displaystyle{\frac{(r+1)^{d-s}}{W_s(\mathbb{S}^d)(r-1)^d}}-1\right)\|\sigma^-_\rho
\|+\|\epsilon_\rho\|/(n-2)}{\|\nu_\rho\|} \nu_\rho  - 1/(n-2)\, \epsilon_\rho.
\end{equation*}
As balayage reduces the norm of a measure (i.e., $\|\sigma^+_\rho \| \leq \|\sigma^+ \|$ and $\|\sigma^-_\rho \| \leq \|\sigma^- \|$), we obtain
\begin{equation*} 
\eta_{\Sigma_\rho, \widetilde{Q}_k, s} \leq \frac{1+\|\sigma^+ \|+\left(\displaystyle{\frac{(r+1)^{d-s}}{W_s(\mathbb{S}^d)(r-1)^d}}-1\right)\|\sigma^-
\|+\|\epsilon_\rho\|/(n-2)}{\|\nu_\rho\|} \nu_\rho  - 1/(n-2)\, \epsilon_\rho.
\end{equation*}
The dominating signed measure at the right-hand side above has total charge
\begin{equation*}
c_{\sigma} \DEF 1+\|\sigma^+ \|+ \left( \frac{(r+1)^{d-s}}{W_s(\mathbb{S}^d)(r-1)^d} - 1 \right) \left\| \sigma^- \right\|
\end{equation*}
(which does not depend on $k$ and $\rho$) and can be rewritten as
\begin{equation}
c_{\sigma} \left( \frac{1 + \| \epsilon_\rho \| / [ c_{\sigma} \, ( n - 2 ) ]}{\| \nu_\rho \|} \, \nu_\rho - \frac{1}{c_{\sigma} \left( n-2 \right)} \, \epsilon_\rho \right),
\end{equation}
where the parenthetical expression is the signed $s$-equilibrium on $\Sigma_\rho$ associated with the Riesz-$s$ external field generated by a point charge of size $q \DEF 1 / [ c_{\sigma} \, ( n - 2 ) ]$ at $\PT{x}_k$. (There is no other dependence on $k$ than that $\PT{x}_k$ determines the axis of symmetry of $\Sigma_\rho$.) It follows from \cite[Theorem~13 with $R = 1$]{BrDrSa2009} (also cf. analysis in \cite[Theorem~1.5]{DrSa2007}) that this signed measure has a negative part for each $0 < \rho < \rho_0$, where the critical $\rho_0$ solves the equation\footnote{Note that here the parameter is the distance of the source (on the sphere) to the boundary of $\Sigma_\rho$. In \cite{BrDrSa2009} the ``altitude'' of the boundary is used.}
\begin{equation*}
q \, \frac{2^{d-s}}{\rho^d} = \Psi_s( \rho) \qquad \text{with} \qquad \Psi_s( \rho ) \DEF W_s( \mathbb{S}^d ) \frac{1 + q \, \| \epsilon_\rho \|}{\| \nu_\rho \|},
\end{equation*}
and $\rho_0$ is also the unique minimizer of $\Psi_s$.
%
%
Hence
\begin{equation*}
q \, \frac{2^{d-s}}{\rho_0^d} = \Psi_s( \rho_0 ) \leq \Psi_s( 0 ) = W_s( \mathbb{S}^d ) \left( 1 + q \right)
\end{equation*}
and it follows that
\begin{equation*}
\left[ \frac{q}{q+1} \, \frac{2^{d-s}}{W_s( \mathbb{S}^d )} \right]^{1/d} \leq \rho_0,
\end{equation*}
where (substituting for $q$)
\begin{equation*}
\frac{q}{q+1} = \frac{1}{c_{\sigma} ( n - 2 ) + 1} \geq \frac{1}{c_{\sigma} \, n} \, \left( 1 - \frac{2 c_{\sigma} - 1}{n} \right)^{-1} > \frac{1}{c_{\sigma} \, n}.
\end{equation*}
The last inequality holds for all $n \geq 2$ such that $0 \leq 2 c_{\sigma} - 1 < n$ and the constant $K_{Q,s}$ takes the form given in \eqref{SepRes2} which reduces to the constant given in \cite[Theorem~1.5]{BrDrSa2009} in the field-free setting. Specialization for $s = d - 1$ gives \eqref{SepRes2b}. This completes the proof of the theorem for $d - 2 < s < d$.
\end{proof}

Next, we derive explicit separation estimates for the limiting case $s = d - 2$ in the field-free setting.

\begin{proof}[Proof of Proposition~\ref{prop:well-separation.min.Riesz.(d-2).energy}]
Let $d \geq 3$ and $s = d - 2$. We proceed similar as in the first part of the Proof of Theorem~\ref{thm:main3}. Let $\PSET_{n,d-2} = \{ \PT{x}_1, \dots, \PT{x}_n \}$ be a Riesz $(d-2)$-energy minimizing $n$-point configuration on $\mathbb{S}^d$. Each point in $\PSET_{n,d-2}$ in turn can be identified with the North Pole $\PT{p}$ and thus defines a Riesz external field
\begin{equation} \label{eq:Q.limiting.case}
Q(\PT{x}) \DEF \frac{q}{|\PT{x}-\PT{p}|^{d-2}}, \qquad \PT{x} \in \mathbb{S}^d,
\end{equation}
where $q = 1 / (n - 2)$, so that  $\PSET_{n,d-2} \setminus \{\PT{p} \}$ is an $(n-1)$-point $(Q,d-2)$-Fekete set on~$\mathbb{S}^d$.
Every $m$-point $(Q,d-2)$-Fekete set on $\mathbb{S}^d$ is contained in the (rotational symmetric) extended support $\widetilde{S}_{Q,d-2}$ of the $(d-2)$-extremal measure $\mu_{Q,d-2}$ on $\mathbb{S}^d$ by Corollary~\ref{cor}. It suffices to show that $\widetilde{S}_{Q,d-2} \subset \{ \PT{x} \in \mathbb{S}^d : | \PT{x} - \PT{p} | \geq K_d / n^{1/d} \}$ for some constant $K_d$ independent of $n$.

Observe, that the set $\widetilde{S}_{Q,d-2}$ is contained in the extended support of the $(d-2)$-extremal measure on a compact subset $E \subset \mathbb{S}^d$ whenever $E$ contains the support $S_{Q,d-2}$ of $\mu_{Q,d-2}$.
Hence, by Theorem~\ref{thm:signsupp.2} and following remark, the set $\widetilde{S}_{Q,d-2}$ is a subset of the support of the positive part of the signed $(d-2)$-equilibrium measure on $E$ associated with $Q$ whenever $S_{Q,d-2} \subset E$.
In fact, Proposition~\ref{prop:1}(b) yields that $S_{Q,d-2} \subset E_M \DEF \{ \PT{x} \in \mathbb{S}^d : Q(\PT{x}) \leq M_n \}$ for some positive $M_n$ depending on $n$. 
(Note that $S_{Q,d-2}$ is a spherical cap by Proposition~\ref{prop:ConnThm}.) 
We deduce that $\widetilde{S}_{Q,s} \subset \supp( \overline{\eta}_\rho^+ )$ of the positive part $\overline{\eta}_\rho^+$ of the signed equilibrium measure $\overline{\eta}_\rho$ on the spherical cap $\Sigma_{\rho} = \{ \PT{x} \in \mathbb{S}^d : | \PT{x} - \PT{p} | \geq \rho \}$ for each $\rho \in (0,\rho_c)$, where $\overline{\eta}_{\rho_c}$ coincides with $\mu_{Q,s}$. The radius $\rho_c$ and its estimate in terms of $n$ gives the desired lower bound for the separation of points in $\PSET_{n,d-2}$.

It can be shown that the $(d-2)$-balayage measures onto the spherical cap $\Sigma_{\rho}$ of the uniform measure $\sigma_d$ and a unit point charge at $\PT{p}$,
\begin{equation} \label{eq:barBal}
\overline{\nu}_\rho \DEF \nu_{\rho,d-2} = \bal_{d-2}(\sigma_d,\Sigma_\rho), \qquad \overline{\epsilon}_\rho \DEF \epsilon_{\rho,d-2} = \bal_{d-2}(\delta_{\PT{p}},\Sigma_\rho),
\end{equation}
exist and both have a component that is uniformly distributed on the boundary of $\Sigma_\rho$.
From \cite[Theorem~15]{BrDrSa2009}, and letting $R \to 1^+$, we obtain that
\begin{equation*}
\overline{\eta}_\rho = \frac{\Psi_{d-2}( \rho )}{W_{d-2}( \mathbb{S}^d )} \, \overline{\nu}_\rho - q \, \overline{\epsilon}_{\rho,d-2}
\end{equation*}
with
\begin{equation*}
\dd \overline{\eta}_\rho( \PT{x} ) = \frac{\Psi_{d-2}( \rho )}{W_{d-2}( \mathbb{S}^d )} \, \dd \sigma_d \big|_{\Sigma_\rho}( \PT{x} ) + \left( 1 - \frac{\rho^2}{4} \right)^{d/2-1} \left( \frac{\Psi_{d-2}( \rho )}{4} \, \rho^d - q \right) \dd \delta_{1-\rho^2/2}(u) \, \dd \sigma_{d-1}(\overline{\PT{x}}),
\end{equation*}
where
\begin{equation*}
\Psi_{d-2}( \rho ) \DEF W_{d-2}( \mathbb{S}^d ) \, \frac{1 + q \, \| \overline{\epsilon}_\rho \|}{\| \overline{\nu}_\rho \|}.
\end{equation*}
Using the formulas in \cite{BrDrSa2009} it is easy to show that $\Psi_{d-2}( \rho )$ has a unique minimum in $(0,2)$ at the critical radius $\rho_c$ (cf. \cite[Theorem~15]{BrDrSa2009}). This $\rho_c$ is the solution of the equation
\begin{equation*}
\frac{\Psi_{d-2}( \rho )}{4} \, \rho^d - q = 0.
\end{equation*}
Therefore,
\begin{equation*}
q \, \frac{4}{\rho_c^d} = \Psi_{d-2}( \rho_c ) \leq \Psi_{d-2}( 0 ) = W_{d-2}( \mathbb{S}^d ) \left( 1 + q \right)
\end{equation*}
and the relations~\eqref{eq:(d-2).separation.estimate} and \eqref{eq:(d-2).separation.estimate.constant} follow after substitution of $q = 1 / (n-2)$.

Let $d = 2$ and $s = \log$. In the logarithmic case the external field in \eqref{eq:Q.limiting.case} is replaced with
\begin{equation*}
Q( \PT{x} ) \DEF q \log \frac{1}{| \PT{x} - \PT{p} |}, \qquad \PT{x} \in \mathbb{S}^d,
\end{equation*}
where $q = 1 / ( n - 2 )$. Again, $\PSET_{n,\log} \setminus \{ \PT{p} \}$ is an $(n-1)$-point $(Q,\log)$-Fekete set on $\mathbb{S}^2$. Similarly as before we are led to the investigation of the signed logarithmic equilibrium problem on spherical caps $\Sigma_\rho$. Theorem~17 in \cite{BrDrSa2009} gives (as $R \to 1^+$) that
\begin{equation*}
\dd \overline{\eta}_{\rho,\log}( \PT{x} ) = \left( 1 + q \right) \dd \sigma_2 \big|_{\Sigma_{\rho}}( \PT{x} ) + \left[ \left( 1 + q \right) \frac{\rho^2}{4} - q \right] \dd \delta_{1-\rho^2/2}(u) \dd \sigma_{d-1}( \overline{\PT{x}} ), \quad ( \sqrt{1 - u^2} \, \overline{\PT{x}}, u ) \in \mathbb{S}^d.
\end{equation*}
Clearly, the signed logarithmic equilibrium $\overline{\eta}_{\rho,\log}$ has a negative boundary charge (uniformly distributed over the boundary of $\Sigma_\rho$) for every $\rho \in (0,\rho_c)$, where at the critical distance $\rho = \rho_c$ this boundary charge vanishes; that is,
\begin{equation*}
\rho_c^2 = \frac{4q}{q+1} = \frac{4}{n-1}.
\end{equation*}
The relation for the separation follows.
\end{proof}

In the third and last part of the Proof of Theorem~\ref{thm:main3} we establish separation bounds for the limiting case $s = d - 2$ and $d \geq 3$ given an external field.

\begin{proof}[Proof of Theorem~\ref{thm:main3}, Part~3]
Let $s = d - 2$ and $d \geq 3$. We proceed as in the first part of the Proof of Theorem~\ref{thm:main3}. Let $\PSET_{n,Q,d-2} = \{ \PT{x}_1, \dots, \PT{x}_n\}$ be an $n$-point $(Q,d-2)$-Fekete set on $\mathbb{S}^d$. Each point in $\PSET_{n,Q,d-2}$ defines an external field by means of 
\begin{equation*}
\widetilde{Q}_k(\PT{x}) \DEF Q(\PT{x}) + \frac{1}{(n-2)|\PT{x}-\PT{x}_k|^{d-2}}, \quad \PT{x} \in \mathbb{S}^d, \qquad k = 1, \dots, n.
\end{equation*}
Again Corollary~\ref{cor} and Proposition~\ref{prop:1} guarantee that there is a spherical cap $\Sigma_{r_k}$ such that $\PSET_{m,\widetilde{Q}_k,d-2} \subset \widetilde{S}_{\widetilde{Q}_k,s} \subset \Sigma_{r_k}$ for every $m$-point $(\widetilde{Q}_k,d-2)$-Fekete set on $\mathbb{S}^d$. Hence, we may study the continuous Riesz $(d-2)$-external field problem on $\Sigma_{r_k}$ and use signed $(d-2)$-equilibria on spherical caps $\Sigma_\rho$ in our analysis. 

The signed $(d-2)$-equilibrium $\eta_{\Sigma_\rho,\widetilde{Q}_k,d-2}$ on a spherical cap $\Sigma_\rho$ exists for any $2>\rho>0$, since (with the notations $\overline{\nu}_\rho \DEF \bal_{d-2}( \sigma_d, \Sigma_\rho )$ and $\overline{\epsilon}_\rho \DEF \bal_{d-2}( \delta_{\PT{x}_k}, \Sigma_\rho)$)
\begin{equation} \label{eq:signedeqQ1a}
\overline{\eta}_{\Sigma_\rho, \widetilde{Q}_k, d-2} = c_k \overline{\nu}_\rho - \bal_{d-2}( \sigma, \Sigma_\rho ) - 1/(n-2)  \, \overline{\epsilon}_\rho,
\end{equation}
where $c_k$ is a normalizing constant such that $\| \eta_{\Sigma_\rho,\widetilde{Q}_k, d-2} \| = 1$. In the limiting case $s = d - 2$, balayage introduces a boundary charge uniformly distributed over the boundary of $\Sigma_\rho$. Indeed, from \cite[Lemmas~33 and 36]{BrDrSa2009}, and letting $R \to 1^+$, we have\footnote{We use the correspondence $\rho^2 = 2( 1 - t )$ in \cite[Lemmas~33 and 36]{BrDrSa2009}.}
\begin{align}
\dd \overline{\nu}_\rho( \PT{x} ) &= \dd \sigma_d \big|_{\Sigma_\rho}( \PT{x} ) + \frac{W_{d-2}( \mathbb{S}^d )}{4} \rho^d \left( 1 - \frac{\rho^2}{4}  \right)^{d/2-1} \dd \delta_{1-\rho^2/2}( u ) \, \dd \sigma_{d-1}( \overline{\PT{x}} ), \\
\dd \overline{\epsilon}_\rho( \PT{x} ) &= \left( 1 - \frac{\rho^2}{4}  \right)^{d/2-1} \dd \delta_{1-\rho^2/2}( u ) \, \dd \sigma_{d-1}( \overline{\PT{x}} ).
\end{align}
Furthermore,
\begin{equation*}
\bal_{d-2}( \sigma, \Sigma_\rho ) = \sigma_\rho^+ - \sigma_\rho^-, \qquad \text{where $\sigma_\rho^+ \DEF \bal_{d-2}( \sigma^+, \Sigma_\rho )$ and $\sigma_\rho^- \DEF \bal_{d-2}( \sigma^-, \Sigma_\rho )$.}
\end{equation*}
Note that $\sigma_\rho^+$ (since it is subtracted in \eqref{eq:signedeqQ1a}) can only have a non-positive contribution to the boundary charge on $\Sigma_\rho$. This leaves $\sigma_\rho^-$. Iterating balayage, we have 
\begin{equation*}
\sigma_\rho^- = \bal_{d-2}( \sigma_0^-, \mathbb{S}^d ) 
\end{equation*}
where $\sigma_0^- \DEF \bal_{d-2}( \sigma^-, \mathbb{S}^d )$. From \eqref{signedeqQ4} we get
\begin{equation}
\sigma_\rho^- \leq \bal_{d-2}\Big( \frac{(r+1)^{2}}{W_{d-2}(\mathbb{S}^d)(r-1)^d} \left\| \sigma^- \right\| \sigma_d(\PT{x}), \Sigma_\rho \Big)  = \frac{(r+1)^{2} \|\sigma^-\|}{W_{d-2}(\mathbb{S}^d)(r-1)^d} \, \overline{\nu}_\rho.
\end{equation}
Hence, using that the $(d-2)$-balayage also does not increase the norm (i.e., $\| \sigma_\rho^+ \| \leq \| \sigma^+ \|$ and $\| \sigma_\rho^- \| \leq \| \sigma^- \|$), we arrive at
\begin{equation*}
\overline{\eta}_{\Sigma_\rho, \widetilde{Q}_k, d-2} \leq \frac{1+\|\sigma^+ \|+\left(\displaystyle{\frac{(r+1)^{2}}{W_{d-2}(\mathbb{S}^d)(r-1)^d}}-1\right)\|\sigma^- \|+\|\overline{\epsilon}_\rho\|/(n-2)}{\|\overline{\nu}_\rho\|} \overline{\nu}_\rho  - 1/(n-2)\, \overline{\epsilon}_\rho.
\end{equation*}
The right-hand side has total charge
\begin{equation*}
\overline{c}_\sigma \DEF 1+\|\sigma^+ \|+\left(\displaystyle{\frac{(r+1)^{2}}{W_{d-2}(\mathbb{S}^d)(r-1)^d}}-1\right)\|\sigma^- \|
\end{equation*}
and can be rewritten as $\overline{c}_\sigma \overline{\eta}_{\rho,0}$, where $\overline{\eta}_{\rho,0}$ is the signed $(d-2)$-equilibrium on $\Sigma_\rho$ for the Riesz-$(d-2)$ external field due to a point charge of size $q = 1 / [ \overline{c}_\sigma ( n - 2 ) ]$ placed at $\PT{x}_k$. Now we may proceed as in the last part of the Proof of Proposition~\ref{prop:well-separation.min.Riesz.(d-2).energy} to conclude that the critical distance $\rho_c$ satisfies
\begin{equation*}
q \, \frac{4}{\rho_c^d} \leq W_{d-2}( \mathbb{S}^d ) \left( 1 + q \right);
\end{equation*}
that is
\begin{equation*}
\rho_c \geq \left[ \frac{q}{q+1} \, \frac{4}{W_{d-2}( \mathbb{S}^d )} \right]^{1/d}, \qquad \frac{q}{q+1} = \frac{1}{\overline{c}_\sigma ( n - 2 ) + 1} \geq \frac{1}{\overline{c}_\sigma \, n}.
\end{equation*}
The result for $s = d - 2$ and $d \geq 3$ follows. In the logarithmic case $s = \log$ and $d = 2$ we follow the same argument utilizing \cite[Lemmas~39 and 41]{BrDrSa2009}. Note that the separation estimate in the logarithmic case is the limit as $s \to 0$ of the separation results for ${0 < s < 2}$.
\end{proof}

In the proof of Theorem~\ref{thm:main4} we will use the following facts about \emph{spherical caps}
\begin{equation*}
C( \PT{x}, r ) \DEF \left\{ \PT{y} \in \mathbb{S}^d : \left| \PT{y} - \PT{x} \right| \leq r \right\} = \left\{ \PT{y} \in \mathbb{S}^d :  \PT{y} \cdot \PT{x} \geq 1 - r^2 / 2 \right\}
\end{equation*}
centered at $\PT{x} \in \mathbb{S}^d$ with radius $r \in (0,2)$ provided in \cite{KuSa1998}; namely
\begin{equation} \label{eq:cap.area.expansion}
\sigma_d( C( \PT{x}, r ) ) = \frac{\omega_{d-1}}{\omega_d} \int_{1-r^2/2}^1 \left( 1 - t^2 \right)^{d/2-1} \dd t = \frac{1}{d} \frac{\omega_{d-1}}{\omega_d} \, r^d + o( r^{d+2} ) \qquad \text{as $r \to 0$}
\end{equation}
and the estimate ($d \geq 2$)
\begin{equation} \label{eq:cap.area.estimate}
\sigma_d( C( \PT{x}, r ) ) \leq \frac{1}{d} \frac{\omega_{d-1}}{\omega_d} \, r^d.
\end{equation}
Here, $\omega_d$ denotes the surface area of $\mathbb{S}^d$ and
\begin{equation} \label{eq:gamma.d}
\gamma_d \DEF \frac{\omega_{d-1}}{\omega_d} = \frac{\gammafcn( (d+1)/2 )}{\sqrt{\pi} \, \gammafcn( d / 2 )}.
\end{equation}
Reference \cite{KuSa1998} also gives
\begin{equation} \label{eq:mod.potential.estimate}
\int_{\mathbb{S}^d \setminus C( \PT{x}, r )} \frac{1}{\left| \PT{x} - \PT{y} \right|^s} \dd \sigma_d( \PT{y} ) = \frac{\gamma_d}{s-d} \, r^{d-s} + \mathcal{R}_{s,d}(r),  
\end{equation}
where $\mathcal{R}_{s,d}(r) = o( r^{d-s} )$ as $r \to 0$. A finer analysis of the integral at the right-hand side gives the following explicit estimates.
\begin{lem} \label{lem:explicit.estimates}
The remainder term in \eqref{eq:mod.potential.estimate} satisfies $\mathcal{R}_{s,d}(r) \leq \frac{\gamma_d}{2} \beta_{s,d} \, r^{2+d-s}$, where
\begin{equation} \label{eq:coeff.beta.s.d}
\beta_{s,d}
=
\begin{cases}
0 & \text{for $d < s \leq 2d$,} \\[1.05em]
\dfrac{s/2 - d}{\left( s - d \right) \left( s - d - 2 \right)} & \text{for $2d < s \leq 2d + 2$,} \\[1.05em]
\dfrac{s/2 - d}{d \left(d+2\right)} & \text{for $s > 2d + 2$.}
\end{cases}
\end{equation}
\end{lem}

\begin{proof}
The Funk-Hecke formula (cf. \cite{Mu1966}) gives
\begin{equation*}
\int_{\mathbb{S}^d \setminus C( \PT{x}, r )} \frac{1}{\left| \PT{x} - \PT{y} \right|^s} \dd \sigma_d( \PT{y} ) = \frac{\omega_{d-1}}{\omega_d} 2^{-s/2} \int_{-1}^{1-r^2/2} \left( 1 - t \right)^{-s/2} \left( 1 - t^2 \right)^{d/2-1} \dd t.
\end{equation*}
The standard substitution $1+t = 2 ( 1 - \frac{r^2}{4} ) v$ and \cite[Eq.~15.6.1]{NIST:DLMF} enables us to obtain an expression in terms of a Gauss hypergeometric function by means of
\begin{align}
\int_{\mathbb{S}^d \setminus C( \PT{x}, r )} \frac{1}{\left| \PT{x} - \PT{y} \right|^s} \dd \sigma_d( \PT{y} )
&= \gamma_d 2^{d-s-1} \left( 1 - \frac{r^2}{4} \right)^{d/2} \int_0^1 \frac{v^{d/2-1} \left( 1 - v \right)^{1-1}}{\left( 1 - \left( 1 - \frac{r^2}{4} \right) v \right)^{1+(s-d)/2}} \, \dd v \notag \\
&= \frac{\gamma_d}{d} 2^{d-s-1} \left( 1 - \frac{r^2}{4} \right)^{d/2} \Hypergeom{2}{1}{1+\frac{s-d}{2},\frac{d}{2}}{1+\frac{d}{2}}{1-\frac{r^2}{4}} \label{eq:potential.deleted.cap} \\
&= \frac{\gamma_d}{d} \, r^{d-s} \left( 1 - \frac{r^2}{4} \right)^{d/2} \Hypergeom{2}{1}{d-\frac{s}{2},1}{1+\frac{d}{2}}{1-\frac{r^2}{4}}, \notag
\end{align}
where the last relation follows from the last linear transformation in \cite[Eq.~15.8.1]{NIST:DLMF}. Using that $( 1 - \frac{r^2}{4} )^{d/2} \leq 1$, we obtain the estimate
\begin{equation*}
\int_{\mathbb{S}^d \setminus C( \PT{x}, r )} \frac{1}{\left| \PT{x} - \PT{y} \right|^s} \dd \sigma_d( \PT{y} ) \leq \frac{\gamma_d}{d} \, r^{d-s} \Hypergeom{2}{1}{d-\frac{s}{2},1}{1+\frac{d}{2}}{1} + \widetilde{\mathcal{R}}_{s,d}(r),
\end{equation*}
where
\begin{equation*}
\widetilde{\mathcal{R}}_{s,d}(r) = \frac{\gamma_d}{d} \, r^{d-s} \left( \Hypergeom{2}{1}{d-\frac{s}{2},1}{1+\frac{d}{2}}{1-\frac{r^2}{4}} - \Hypergeom{2}{1}{d-\frac{s}{2},1}{1+\frac{d}{2}}{1} \right).
\end{equation*}
Note that the hypergeometric function at argument $1$ evaluates as $d / ( s - d )$ by \cite[Eq.~15.4.20]{NIST:DLMF}. Setting $f(x) = \Hypergeom{2}{1}{d-s/2,1}{1+d/2}{1-x}$, the mean-value theorem and the differentiation formula \cite[Eq.~15.5.1]{NIST:DLMF} gives
\begin{equation*}
\begin{split}
\frac{f(r^2/4)-f(0)}{r^2/4-0} = f^\prime( \xi )
&= - \frac{d-\frac{s}{2}}{1+\frac{d}{2}} \, \Hypergeom{2}{1}{1+d-\frac{s}{2},2}{2+\frac{d}{2}}{1-\xi} \\
&= - \frac{d-\frac{s}{2}}{1+\frac{d}{2}} \, \frac{\gammafcn(2+d/2)}{\gammafcn(2) \gammafcn(d/2)} \int_0^1 \frac{t \left( 1 - t \right)^{d/2-1}}{\left( 1 - \left( 1 - \xi \right) t \right)^{1+d-s/2}} \dd t
\end{split}
\end{equation*}
for some $\xi \in (0, r^2/4 )$. (The last step follows again from \cite[Eq.~15.6.1]{NIST:DLMF}.) The result follows by considering the integral (which reduced to a Beta function integral) for the three given cases $d < s \leq 2d$, $2d < s < 2d+2$ and $s > 2d + 2$ and the observation that $\mathcal{R}_{s,d}(r) \leq \widetilde{\mathcal{R}}_{s,d}(r)$.
\end{proof}

\begin{proof}[Proof of Theorem~\ref{thm:main4}]
We follow the argument in \cite[Corollary 4]{KuSa1998}. Let $X_n = \{ \PT{x}_1, \dots, \PT{x}_n \}$ be a $(Q,s)$-Fekete set on $\mathbb{S}^d$. Set $D_k \DEF B \setminus C( \PT{x}_k, \rho_n )$ for $k = 1, \dots, n$, where $B \subset \mathbb{S}^d$ is such that $\sigma_d( B) > 0$ and $\int_B | Q( \PT{x} ) | \dd \sigma_d( \PT{x} ) < \infty$, and where $\rho_n \DEF ( \sigma_d( B) / n )^{1/d}$. Setting $D \DEF \cap_{k=1}^n D_k$, we have from \eqref{eq:cap.area.estimate} that
\begin{equation} \label{eq:sigma.d.D.estimate}
\sigma_d( D ) \geq \sigma_d( B) - n \sigma_d( C( \cdot, \rho_n ) ) \geq \sigma_d( B) \left( 1 - \frac{1}{d} \gamma_d \right) > 0.
\end{equation}
Consider, for a given index $j$, the function
\begin{equation*}
H_j( \PT{x} ) \DEF \sum_{\substack{k = 1 \\ k \neq j}}^n \frac{1}{\left| \PT{x} - \PT{x}_k \right|^{s}} + 2 \left( n - 1 \right) Q( \PT{x} ), \qquad \PT{x} \in \mathbb{S}^d.
\end{equation*}
As $X_{n,Q,s}$ minimizes the discrete weighted $s$-energy \eqref{eq:discrete.Riesz.energy} associated with $Q$, we have
\begin{equation}
\begin{split} \label{eq:pre.master.inequality}
\sigma_d(D) \, H_j( \PT{x}_j )
&\leq \int_D H_j( \PT{x} ) \, \dd \sigma_d( \PT{x} ) \\
&= \sum_{\substack{k = 1 \\ k \neq j}}^n \int_D \frac{1}{\left| \PT{x} - \PT{x}_k \right|^{s}} \dd \sigma_d( \PT{x} ) + 2 \left( n - 1 \right) \int_D Q( \PT{x} ) \, \dd \sigma_d( \PT{x} ).
\end{split}
\end{equation}
Using the inclusions $D \subset D_k = B \setminus C( \PT{x}_k, \rho_n ) \subset \mathbb{S}^d \setminus C( \PT{x}_k, \rho_n )$ and \eqref{eq:mod.potential.estimate}, we have
\begin{equation*}
\int_D \frac{1}{\left| \PT{x} - \PT{x}_k \right|^{s}} \dd \sigma_d( \PT{x} )
\leq \int_{\mathbb{S}^d \setminus C( \PT{x}_k, \rho_n )} \frac{1}{\left| \PT{x} - \PT{x}_k \right|^{s}} \dd \sigma_d( \PT{x} ) = \frac{1}{s-d} \gamma_d \, (\rho_n)^{d-s} + \mathcal{R}_{s,d}( \rho_n ),
\end{equation*}
where $\mathcal{R}_{s,d}( \rho_n ) = o( (\rho_n)^{d-s} )$ as $\rho_n \to 0$; in fact, by Lemma~\ref{lem:explicit.estimates},
\begin{equation*}
\frac{n}{\sigma_d(D)} \, \mathcal{R}_{s,d}( \rho_n ) \leq \frac{\gamma_d}{2} \beta_{s,d} \, \frac{\sigma_d( B)}{\sigma_d(D)} \, (\rho_n)^{2 - s}.
\end{equation*}
Dividing through by $\sigma_d( D )$ in \eqref{eq:pre.master.inequality}, subtracting off the constant $Q( \PT{x}_j )$ (if it is negative) from both sides of the new inequality and substituting the estimates above, we arrive at
\begin{equation*}
\begin{split} 
H_j( \PT{x}_j ) - 2 \left( n - 1 \right) \min\{ 0, Q( \PT{x}_j ) \}
&\leq \frac{\gamma_d}{s-d} \frac{\sigma_d( B)}{\sigma_d(D)} \, (\rho_n)^{- s} + \frac{\gamma_d}{2} \beta_{s,d} \, \frac{\sigma_d( B)}{\sigma_d(D)} \, (\rho_n)^{2 - s} \\
&\phantom{=}+ \frac{2\left( n - 1 \right)}{\sigma_d( D )} \int_D \big( Q( \PT{x} ) - \min\{ 0, Q( \PT{x}_j ) \} \big) \dd \sigma_d( \PT{x} ).
\end{split}
\end{equation*}
Note that $\min\{ 0, Q( \PT{x}_j ) \}$ is well-defined and finite by the lower semi-continuity of $Q$ and minimality of $X_{n,Q,s}$.
Since $\left| \PT{x}_j - \PT{x}_k \right|^{-s}$ ($k = 1, \dots, n$, $k \neq j$) is a lower bound for the left-hand side, we have
\begin{equation}
\begin{split} \label{eq:master.inequality}
\left| \PT{x}_j - \PT{x}_k \right|^{-s}
&\leq \frac{\sigma_d( B)}{\sigma_d(D)} \, (\rho_n)^{- s} \Bigg\{ \frac{\gamma_d}{s-d}  + \frac{\gamma_d}{2} \beta_{s,d} \, (\rho_n)^{2} \\
&\phantom{=}+ 2\left( n - 1 \right) (\rho_n)^{s} \frac{1}{\sigma_d( B )} \int_D \big( Q( \PT{x} ) - \min\{ 0, Q( \PT{x}_j ) \} \big) \dd \sigma_d( \PT{x} ) \Bigg\}.
\end{split}
\end{equation}
Integrability (with respect to $\sigma_d$) of $Q$ over $B \supset D$ and $\min\{ 0, Q( \PT{x}_j ) \} \geq \min\{ 0, \underline{M}^Q \}$, where $\underline{M}^Q$ is the minimum of $Q$ over $\mathbb{S}^d$, yield
\begin{equation} \label{eq:master.inequality.estimate}
\begin{split}
&\left| \frac{1}{\sigma_d( B )} \int_D \big( Q( \PT{x} ) - \min\{ 0, Q( \PT{x}_j ) \} \big) \dd \sigma_d( \PT{x} ) \right| \\
&\phantom{equalsequals}\leq \frac{1}{\sigma_d( B)} \int_{B} \left| Q( \PT{x} ) \right| \dd \sigma_d( \PT{x} ) - \min\{ 0, \underline{M}^Q \} < \infty.
\end{split}
\end{equation}
Hence, the result follows from \eqref{eq:master.inequality}, \eqref{eq:sigma.d.D.estimate} and the definition of $\rho_n$.
\end{proof}


\end{document}